\documentclass[a4paper,12pt]{amsart}

\usepackage{marginnote,geometry,float}
\usepackage{dsfont,amsmath,mathrsfs,comment}
\usepackage[usenames,dvipsnames,svgnames,table]{xcolor}
\usepackage{cancel,amssymb}
\usepackage{graphicx}
\usepackage[pdfdisplaydoctitle=true,
            colorlinks=true,
            urlcolor=blue,
            citecolor=blue,
            linkcolor=blue,
            pdfstartview=FitH,
            pdfpagemode= UseNone,
            bookmarksnumbered=true]{hyperref}
\usepackage[color=blue!20, linecolor=red!]{todonotes}
\usepackage{lineno}
\usepackage{pb-diagram}
\newtheorem{theorem}{Theorem}[section]

\newtheorem{proposition}[theorem]{Proposition}

\theoremstyle{definition}

\theoremstyle{remark}

\newtheorem*{example}{Example}

\numberwithin{equation}{section}
\newcommand{\ZZ}{\mathbb{Z}} 
\newcommand{\RR}{\mathbb{R}} 
\newcommand{\CC}{\mathbb{C}} 

\newcommand{\id}{\mathds{1}}

\newcommand{\Ss}[1]{W^{k,p}(\mathbb{S}^1)}

\newcommand{\abs}[1]{\left\vert#1\right\vert}

\newcommand{\bra}[1]{\left\langle {#1}\right\vert}
\newcommand{\ket}[1]{\left\vert{#1}\right\rangle}

\newcommand{\todotony}[1]{\todo[inline, color=blue!20!green!30]{#1}}

\numberwithin{equation}{section}

\title[A dressing method for the Camassa-Holm equation]{Camassa-Holm cuspons, solitons and their interactions via the dressing method}

\author{Rossen Ivanov}
\address{School of Mathematical Sciences,Technological University Dublin, City Campus, Kevin Street, Dublin D08 NF82, Ireland}
\email{rossen.ivanov@dit.ie}

\author{Tony Lyons}
\address{Department of Computing and Mathematics, Waterford Institute of Technology, Waterford, Ireland}
\email{tlyons@wit.ie}

\author{Nigel Orr}
\address{School of Mathematical Sciences,Technological University Dublin, City Campus, Kevin Street, Dublin D08 NF82, Ireland}
\email{nigel.orr@dit.ie }

\begin{document}
\begin{abstract}
A dressing method is applied to a matrix Lax pair for the Camassa-Holm equation, thereby allowing for the construction of several global solutions of the system. In particular solutions of system of soliton and cuspon type are constructed explicitly. The interactions between soliton and cuspon solutions of the system are investigated. The geometric aspects of the Camassa-Holm equation ar re-examined in terms of quantities which can be explicitly constructed via the inverse scattering method.
\end{abstract}

\maketitle

\section{Introduction}
This paper aims to explai how the dressing method, well known in the soliton theory, can be applied to one of the most famous integrable equations of the last 20 years - the Camassa-Holm (or CH) equation. In particular, the method allows for the explicit construction of the soliton and cuspon solutions and for further investigation of the interactions between them. The Camassa-Holm equation is given by
\begin{equation}\label{eq:ch}
\left\{
\begin{aligned}
    q_t&+ 2u_x q + uq_x=0\\
    q&=u-u_{xx}
\end{aligned}
\right.\tag{CH}
\end{equation}
and in the following we impose the boundary conditions $\displaystyle{\lim_{\abs{x}\to\infty }}u(x,t)=u_0$, where $u_0>0$  is constant.
The CH equation admits both smooth and peaked travelling wave solutions cf. \cite{CH93,CHH94}. Physically the Camassa-Holm equation has attracted a great deal of interest as an approximate fluid model for two-dimensional water waves propagating over a flat bed \cite{CH93,CHH94,F95,DGH03,DGH04,J02,J03,J03a,CL09,HI11} as well as in other set-ups \cite{titi}. Interpreted as a fluid model, the solutions $u(x,t)$ represent the fluid particle velocity induced by the passing wave, or alternatively as the surface elevation associated with the wave.  Moreover, the system also constitutes a model for the propagation of nonlinear waves in  cylindrical hyper-elastic rods, in which the solutions $u(x,t)$ represent the radial stretching of a rod relative to the undisturbed state, see \cite{Dai98}.

For the past number of decades the Camassa-Holm equation has proven to be a remarkably fertile field of mathematical research, with  the volume of research papers dedicated to various aspects of the system most likely measured in thousands, and as such our bibliography is by no means exhaustive. A particularly striking feature of the Camassa-Holm equation relates to the existence of peaked solutions for the system, which are solutions of the form
\begin{equation}
    u(x,t) = q_0e^{-\abs{x-p_0t}}\quad \text{where }\lim_{\abs{x}\to\infty }u(x,t)=0,
\end{equation}
with $q_0$ and $p_0$ being constants. These peaked solutions (or peakons) are  weak solutions whose wave crests appear as peaks, see \cite{CH93,CHH94,CS00,M2}. In addition the Camassa-Holm equation also allows for the existence of breaking wave solutions, which are realised as solutions which remain bounded but whose gradient becomes unbounded in a finite time, cf. \cite{CH93, CHH94, CE98, CE98a,C00,PS}. The presence of both peaked and breaking wave solutions for the system \eqref{eq:ch} ensures the Camassa-Holm equation is a highly interesting physical model. To compliment the utility of the system in modelling a diversity of physical phenomena, the Camassa-Holm equation exhibits a rich mathematical structure. The equation is a member of a bi-Hamiltonian hierarchy of equations \cite{F80} and it is integrable with a Lax pair representation \cite{CH93}. A notable property of the system is its formulation as a geodesic flow on the Bott-Virasoro group \cite{M98,HMR98,HMR-PRL,CK03}.

Soliton solutions of the Camassa-Holm equation have been derived by manifold methods, including but not restricted to Hirota's method \cite{M1,M2,PI,PII,PIII}, via the B\"acklund transform method \cite{RS,Li04}, along with the inverse scattering method \cite{CGI,BMS2}. In the current work we develop a modified version of the inverse scattering method, namely the dressing method, to construct the cuspon and soliton solutions and cuspon-soliton interactions of the Camassa-Holm system.
The dressing method is an efficient variation of the inverse scattering transform which allows for a very direct construction of soliton solutions of integrable PDE \cite{ZS1,ZS2,ZMNP,GVY}. The essential procedure behind this dressing method is the construction of a nontrivial (dressed) eigenfunction of an associated spectral problem from the known (bare) eigenfunction, by means of the so-called dressing factor. This dressing factor is analytic in the entire complex plane (of values of the spectral parameter), except for a collection simple poles at pre-assigned discrete eigenvalues. This bare spectral problem is obtained for some trivial solution, e.g. $u(x,t)= u_0,$ where $u_0$ as indicated above is the asymptotic value of the solution, which we require to be constant and strictly positive. Since the potential terms of this bare spectral problem are simply constant, this means the spectral problem is readily solved to yield the bare eigenfunction. In the following we will outline the construction of the solutions of the CH equation associated with the discrete spectrum of the Lax operator, i.e. the solitons and cuspons.

The Camassa-Holm equation has many similarities with the integrable Degasperis-Procesi (or DP) equation \cite{DP,DHH}. The inverse scattering transform of the DP equation is studied in \cite{CIL,BMS1}, and in particular the dressing method for the DP equation is presented in \cite{CI16}. In the following we will reformulate equation \eqref{eq:ch} in the form of a matrix Lax pair, and impose an appropriate gauge transformation on this matrix equation, which reduces the spectral problem to the familiar Zakharov-Shabat spectral problem. We deduce several important reduction symmetries of the spectral problem, which are then utilised in constructing the dressing factor. This in turn allows for the construction of solutions of the dressed spectral problem, from solutions of the bare spectral problem which are readily solved. Finally, using these dressed eigenfunctions, we obtain the physical solutions we seek by solving a straight forward differential equation.

\section{The Spectral Problem for the Camassa-Holm Equation}\label{sec2}
\subsection{From the scalar to the matrix Lax pair}\label{sec2.1}
The following spectral problem
\begin{equation}\label{sec2.1eq1}
\left\{
\begin{aligned}
    &\phi_{xx}=\left(\frac{1}{4}+\lambda^2 q\right)\phi\\
    &\phi_{t}=\left(\frac{1}{2\lambda^2}-u\right)\phi_{x}+\frac{u_{x}}{2}\phi.
\end{aligned}
\right.
\end{equation}
may be seen to represent the Camassa-Holm equation (cf. equation \eqref{eq:ch}), by imposing the compatibility condition $\phi_{xxt}\equiv\phi_{txx}$ on the spectral function $\phi$ and comparing terms of equal order in $\lambda,$ \cite{CH93}. The constant $\lambda$ appearing in equation \eqref{sec2.1eq1} is the time-independent spectral parameter, while the potential $u(x,t)$ corresponds to a solution of the CH equation. Solutions  of the Camassa-Holm equation may be obtained from the spectral problem above by means of the Inverse Scattering Transform, and the reader is referred to the works \cite{CGI,CGI2} for further discussion in this regard. A discussion of the Inverse Scattering Transform applied to constructing periodic solutions of the system \eqref{eq:ch} may be found in \cite{CM99, C98,GH03}. We note that in contrast to some previous works we shall omit the dispersion term $u_x$ in the CH equation, and instead we will allow for a constant asymptotic value $u(x,t)\to u_0 >0$ as $\abs{x}\to\infty$. This is the setup adopted in \cite{RS}.

Suppose that $u(\cdot,t)-u_0$ is a Schwartz class function for all $t$, while the initial data is chosen such that $q(x,0) > 0$. Symmetry of the Camassa-Holm equation then ensures that $q(x,t) > 0$ for all $t$, cf. \cite{C01}. Letting $k^{2}=-\frac{1}{4}-\lambda^2 u_0$, the spectral parameter may be written as
\begin{equation} \label{lambda}
\lambda^2(k)= -\frac{1}{u_0}\left(k^{2}+\frac{1}{4}\right),
\end{equation}
and the reader is referred to \cite{C01} for a discussion of the spectrum of the problem formed by equations \eqref{sec2.1eq1}--\eqref{lambda}. Then the continuous spectrum in terms of $k$
corresponds to $k\in\RR$. The discrete spectrum (corresponding to $k\in\mathbb{C}_+$--the upper half-plane) consists of a finite number of points $k_{n}=i\kappa _{n}$,
$n=1,\ldots,N$ where $\kappa_n$ is real and $0<\kappa_{n}<1/2$, with the corresponding spectral parameter $\lambda_n=\lambda(i\kappa_n)$ being purely imaginary. Moreover for any $\kappa_n$ there are two such eigenvalues, denoted by $\lambda_n=\pm i \omega_n$ where  $\omega_n>0.$

We note that a discrete eigenvalue with $\kappa_n >1/2$ (if at all possible) would not lead to solutions $u(\cdot,t)-u_0$ from the Schwartz class. In such case $\lambda_n$ is real. Later we will find out that indeed this choice corresponds to solutions with a cusp at the crest (cuspons), which are clearly outside of the Schwartz class functions.

To implement the dressing method it is first necessary to reformulate the spectral problem in \eqref{sec2.1eq1} as a \textit{matrix Lax pair}. To achieve this we let $\phi_1$ denote an eigenfunction of equation \eqref{sec2.1eq1}, and  we observe that the first member of this spectral problem may be reformulated as
\begin{equation}\label{sec2.1eq2}
  \left(\partial_x-\frac{1}{2}\right)\left(\partial_x+\frac{1}{2}\right)\phi_{1}=\lambda^2 q\phi_{1}.
\end{equation}
The second member of equation \eqref{sec2.1eq1} may be re-written in terms of the following auxiliary spectral function
\begin{equation} \label{sec2.1eq3}
    \phi_{2}:=\frac{1}{\lambda}\left(\partial+\frac{1}{2}\right)\phi_{1}
\end{equation}
from which we immediately deduce that
\begin{equation*}
    \left(\partial_x-\frac{1}{2}\right)\phi_2 = \lambda q \phi_{1},
\end{equation*}
having imposed equation \eqref{sec2.1eq2}. Defining the eigenvector
\[\Phi=\left(\begin{array}{cc}\phi_1\\ \phi_2\end{array}\right),\]
we reformulate the spectral problem \eqref{sec2.1eq1} according to
\begin{equation}\label{sec2.1eq4}
\left\{
\begin{aligned}
&\Phi_{x}=\mathcal{L}\Phi\qquad \mathcal{L}:=\left(\begin{array}{cc}-\frac{1}{2} &\lambda \\ \lambda q & \frac{1}{2}\end{array}\right)\\
&\Phi_t=\mathcal{M}\Phi\qquad \mathcal{M}:=\left(\begin{array}{cc}\frac{1}{2}(u+u_x)-\frac{1}{4\lambda^2} &\frac{1}{2\lambda}-\lambda u\\\frac{1}{2\lambda}(q+u_x+u_{xx})-\lambda u q & \frac{1}{4\lambda^2}-\frac{1}{2}(u+u_x)\end{array}\right)
\end{aligned}
\right.
\end{equation}
which constitutes a matrix Lax pair for the Camassa-Holm equation. Moreover, $\mathcal{L}$, $\mathcal{M}$ take values in the $\frak{sl}(2)$ algebra, thus $\Phi$ belongs to the corresponding group, $SL(2)$.

The compatibility condition $\Phi_{tx}\equiv\Phi_{xt}$ for every eigenvector $\Phi$ immediately implies the \textit{zero-curvature condition}, namely
\begin{equation}\label{sec2.1eq5}
    \mathcal{L}_{t}-\mathcal{M}_{x}+\left[\mathcal{L},\mathcal{M}\right]=0,
\end{equation}
where the bi-linear operator $\left[\cdot,\cdot\right]$ denotes the usual matrix commutator. As with the scalar formulation of the spectral problem, comparison of terms  of equal order in the spectral parameter $\lambda$ within the zero-curvature condition yields
\begin{description}
  \item[$\mathcal{O}(\lambda^{0})$] $u-u_{xx} =q$
  \item[$\mathcal{O}(\lambda^{1})$] $q_t+2u_x q+u q_x=0$,
\end{description}
which is precisely the Camassa-Holm equation.


\subsection{The gauge transformed $SL(2)$ spectral problem}\label{sec2.2}
In the following, the dressing method will be implemented on a \text{gauge equivalent} matrix-valued eigenfunction $\Psi \in SL(2)$, defined as follows
\begin{equation}\label{sec2.2eq1}
    \Phi=:G\Psi,
\end{equation}
where the gauge transformation $G \in SL(2)$ is given by
\begin{equation}\label{sec2.2eq2}
    G = \left(\begin{array}{cc}
                q^{-\frac{1}{4}} & 0\\
                0 & q^{\frac{1}{4}}
                \end{array}\right).
\end{equation}
Making this replacement in the spectral problem \eqref{sec2.1eq4}, the gauge equivalent spectral problem for $\Psi$ is given by
\begin{equation}\label{sec2.2eq3}
\begin{aligned}
&\Psi_{x} = \tilde{L}\Psi\qquad \Psi_{t} =\tilde{M}\Psi,
\end{aligned}
\end{equation}
where we denote
\begin{equation}\label{sec2.2eq4}
\tilde{L}:=G^{-1}\mathcal{L}G-G^{-1}G_{x}\qquad\tilde{M}:=G^{-1}\mathcal{M}G-G^{-1}G_{t}.
\end{equation}
In particular we find that the equation for $\Psi$ may be written as
\begin{equation}\label{sec2.2eq5}
\left\{
\begin{aligned}
&\Psi_x + (\tilde{h} \sigma_{3}-\lambda \sqrt{q} J) \Psi(x,t,\lambda)=0\\
&\tilde{h} = \frac{1}{2}-\frac{q_x}{4q},\qquad J=\left(\begin{matrix}0&1\\1&0\end{matrix}\right)\qquad \sigma_3=\left(\begin{matrix}1&0\\0&-1\end{matrix}\right).
\end{aligned}
\right.
\end{equation}
We note that this spectral problem appears to be ``energy dependent'' since the potential appears in combination with the spectral parameter in the off diagonal terms. However, introducing the re-parameterisation
\begin{equation}\label{y}
dy=\sqrt{q}dx, \qquad y=y(x,t),
\end{equation}
the spectral problem acquires the form of the standard Zakharov-Shabat spectral problem
\begin{equation}\label{ZS}
\begin{cases}
&\Psi_y+ L(\lambda) \Psi(y,t,\lambda) =0,\\
&L(\lambda)=h \sigma_3 -\lambda J\\
&h= \frac{1}{2\sqrt{q}}-\frac{q_y}{4q},
\end{cases}
\end{equation}
and the reader is referred to \cite{ZS1,ZS2,ZMNP,GVY} for further discussion concerning such spectral problems. Since $L(\lambda)$ takes values in the Lie algebra $\mathfrak{sl}(2)$ it follows that the eigenfunctions take values in the corresponding Lie group - $SL(2).$
From \eqref{y} one can write $x=X(y,t),$ which gives the parametric representation of $x$ for given $t$. This is a very important object in what follows due to the fact that the solution can be expressed through $X(y,t).$ We note that when $x \to \infty,$  asymptotically $q \to u_0$. Then in the view of \eqref{y} it is natural to expect that when $y \to \infty,$ $X\to \frac{y}{\sqrt{u_0}} +\text{const}.$
\begin{theorem}
Suppose that when $y\to \infty$ we have  $X_y\to \frac{1}{\sqrt{u_0}} $ and $X_t\to u_0$. Then the solution in parametric form can be represented as $u(X(y,t),t)=X_t(y,t).$
\end {theorem}
\begin{proof}
We can write \eqref{eq:ch} in the form $$\partial_t \sqrt{q(X(y,t),t)}+\partial_X (\sqrt{q(X,t)} u(X,t))=0,$$ where $X$ depends on $y$ and $t$. We reformulate these derivative in terms of $(y,t)$-varibles and noting that $$\frac{d}{dt}=\frac{\partial}{\partial t}+ X_t \frac{\partial}{\partial X}, \qquad \frac{\partial}{\partial X}=\frac{1}{X_y}\frac{\partial}{\partial y}, \qquad  \sqrt{q(X(y,t),t)}=\frac{1}{X_y(y,t)}$$
we find that
$$\frac{d}{dt} \sqrt{q(X(y,t),t)} - X_t {\partial_X}\sqrt{q(X(y,t),t)}+ \partial_X (\sqrt{q(X,t) } u(X,t))=0,$$ and with some algebra this gives $\left( \frac{u(X,t)-X_t}{X_y}\right)_y=0.$ Thus $u(X,t)=X_t + F(t)X_y$ for some function $F(t)$.  The boundary conditions when $y \to \infty$ give $F(t)\equiv 0$.
\end{proof}

\subsection{Diagonalisation}\label{sec2.3}
Imposing the trivial solution $u(x,t)\equiv u_0$ on the re-parameterised spectral problem, the the so-called \textit{bare} spectral problem emerges, given by
\begin{equation}\label{sec2.3eq1}
\begin{cases}
&\Psi_{0,y}+(h_0\sigma_3-\lambda J)\Psi_0=0,\\
&\Psi_{0,t}-\frac{1}{2h_0}\left(u_0-\frac{1}{2\lambda^2}\right ) (h_0\sigma_3-\lambda J)\Psi_0=0,\\
&h_0=\frac{1}{2\sqrt{u_0}}.
\end{cases}
\end{equation}
Since $dy = \sqrt{u_0}dx$ then $y$ is simply a re-scaling of $x$ for this bare spectral problem. The solutions of this linear system are readily obtained, and found to be of the form
\begin{equation}\label{Psi0}
\Psi_0(y,t,\lambda)= V(\lambda) e^{-\sigma_3 \Omega(y,t,\lambda)}V^{T}(\lambda) C,
\end{equation}
where $C$ is an arbitrary constant matrix and
\begin{equation}\label{L-V}
\left\{
\begin{aligned}
&\Omega(y,t,\lambda)=\Lambda(\lambda) \left(y-\frac{1}{2h_0} \left(u_0 -\frac{1}{2\lambda^2}\right)t \right)\\
&\Lambda(\lambda)=\sqrt{h_0^2 + \lambda^2}\\
&V(\lambda)=\left(\begin{array}{cc}
        \cos \theta & -\sin \theta \\ \sin \theta & \cos \theta \\
\end{array}\right) \\
& \cos \theta =\sqrt{\frac{\Lambda+h_0}{2\Lambda}}, \qquad \sin \theta =\sqrt{\frac{\Lambda-h_0}{2\Lambda}}.
\end{aligned}
\right.
\end{equation}
In the following, the spectral parameter $\lambda$ will be restricted to ensure $\Lambda$ is always real and positive, however this in turn will mean $\theta$ may be either real or imaginary.

\subsection{Symmetry reductions of the spectral problem}\label{sec2.4}
It is easily verified that the spectral operator $L(\lambda)= h\sigma_3-\lambda J$ appearing in equation \eqref{ZS}  possesses the following $\ZZ_2$-symmetry reductions:
\begin{equation}\label{sec2.4eq1}
\begin{split}
       & \sigma_{3}\bar{L}(-\bar{\lambda})\sigma_3 = {L}(\lambda) \\
        & \bar{L}(\bar{\lambda})=L(\lambda)
\end{split}
\end{equation}
where $\bar{\lambda}$ denotes the complex-conjugate of $\lambda.$ This reduction is found to arise due to the symmetry relation $\sigma_3J\sigma_3=-J$. Moreover, the same $\ZZ_2$-symmetry reduction is observed by the associated operator $M(\lambda)$. A crucial result of the symmetry relation \eqref{sec2.4eq1} is that the potential $h(y,t)$ must be real.
Furthermore, since all aspects of the spectral problem \eqref{L-V} obey the  reductions in \eqref{sec2.4eq1}, the associated solutions $\Psi(y,t,\lambda)$ and the dressing factor $g(y,t,\lambda)$ (cf. Section \ref{sec3}) observe the following reductions:
\begin{equation}\label{Psisym1}
\begin{split}
&\sigma_{3}\bar{\Psi}(y,t,-\bar{\lambda})\sigma_3 = {\Psi}(y,t,\lambda), \\
&\bar{\Psi}(y,t,\bar{\lambda}) = {\Psi}(y,t,\lambda).
\end{split}
\end{equation}
Moreover, noting that
\begin{equation}\label{sec2.4eq4}
J^2=\mathds{1},\qquad J\sigma_3J = -\sigma_3,
\end{equation}
while also using $\Psi^{-1}(x,t,\lambda){\Psi}(x,t,\lambda)=\mathds{1}$, we observe that
\begin{equation}\label{sec2.4eq6}
\Psi^{-1}_y(\lambda) = \Psi^{-1}(\lambda)L(\lambda) \Rightarrow \Psi^{-1}_y(\lambda)^{T} = L(\lambda)\Psi^{-1}(\lambda)^{T},
\end{equation}
having used $L^{T}(\lambda)=L(\lambda)$ in the last equation. Hence with equation \eqref{sec2.4eq4} we deduce
\begin{equation}\label{sec2.4eq7}
        \left(J\Psi^{-1}(\lambda)^{T}J\right)_y + L(-\lambda)\left(J\Psi^{-1}(\lambda)^{T}J\right)=0,
\end{equation}
that is to say $\Psi(-\lambda)$ and $J\Psi^{-1}(\lambda)^{T}J$ satisfy the same spectral problem, the solutions of which are unique (when fixed by the corresponding asymptotics in $y$ and $\lambda$), and thus
\begin{equation}\label{sec2.4eq8}
\Psi^{-1}(y,t,\lambda) = J\Psi(y,t,-\lambda)^{T}J.
\end{equation}

\section{The Dressing Method}\label{sec3}
\subsection{The dressing factor}\label{sec3.1}
The soliton, cuspon and soliton-cuspon solutions have associated discrete spectra containing a finite number distinct eigenvalues $\left\{\lambda_n\right\}_{n=1}^{N}$,
with the eigenfunctions of the spectral problem \eqref{ZS} being singular at these discrete eigenvalues. Starting from a trivial (or bare) solution $u(x,t)=u_0$, where $u_0$ is constant,
with its associated eigenfunction $\Psi_0(x,t,\lambda)$, we may obtain an eigenfunction $\Psi(x,t,\lambda)$ corresponding to soliton solutions, via the {\it dressing factor} $g(x,t,\lambda)$, defined by the following
\begin{equation}\label{sec3.1eq1}
        \Psi(x,t,\lambda) = g(x,t,\lambda) \Psi_0(x,t,\lambda).
\end{equation}
The dressing factor $g(x,t,\lambda)$ is singular at each $\lambda=\lambda_n$ belonging to the discrete spectrum, and is otherwise analytic for $\lambda\in\CC_+$.

In the following we shall work with the $y$-representation introduced in equation \eqref{ZS}, which we implement via the diffeomorphism $x=X(y,t)$.
Under this representation the dressing factor then satisfies
\begin{equation}\label{sec3.1eq3}
        \partial_y g +h\sigma_3 g - gh_0 \sigma_3 -\lambda[J,g]=0,
\end{equation}
where $g(y,t, \lambda)$ is to be interpreted as $g(X(y,t),t,\lambda)$. Moreover, since the solution $\Psi(X(y,t),t,\lambda)\equiv \Psi(y,t,\lambda)$ belongs to the Lie group $SL(2)$, the  factor $g\in SL(2)$  and also satisfies the reductions given by equations  \eqref{Psisym1} and \eqref{sec2.4eq8}, namely
\begin{equation}\label{gsym}
\begin{split}
        & \sigma_{3}\bar{g}(y,t,-\bar{\lambda})\sigma_3 = {g}(y,t,\lambda)  \quad \text{or} \quad \bar{g}(y,t,\bar{\lambda}) = {g}(y,t,\lambda), \\
&g^{-1}(y,t,\lambda) = Jg(y,t,-\lambda)^{T}J.
\end{split}
\end{equation}

The physical solutions $u(x,t)$ are extracted from the associated spectral functions $\Psi(x,t,\lambda=0)$, which we evaluate via the spectral problem \eqref{sec2.2eq5}  at $\lambda=0$. These solutions of the spectral problem are of the form
\begin{equation}\label{sec3.2eq1}
        \Psi(X(y,t),t,0)= e^{-\frac{1}{2}\left(X-\ln\sqrt{q}\right)\sigma_3}.
\end{equation}
Meanwhile, re-parameterising in terms of the $y$-variable, the eigenfunction of the dressed spectral problem at $\lambda=0$ (cf. equation \eqref{ZS}) may be written as
\begin{equation}\label{sec3.2eq3}
\Psi(X(y,t),t,0) = g(y,t,0)\Psi_{0}(y,t,0)K_0,
\end{equation}
where $\Psi_{0}(y,t,0)$ is a solution of the bare spectral problem when $\lambda=0$. Hence the matrix $K_0\in SL(2)$ is an arbitrary constant matrix, a consequence of equation \eqref{sec3.1eq1}.
We note however that the $t$-dependence of the bare spectral function $\Psi_{0}(y,t,\lambda)$ becomes singular when $\lambda=0$, an observation which is immediately obvious when we refer to the second equation of the spectral problem \eqref{sec2.3eq1}. As such, in equation \eqref{sec3.2eq3} we only consider the time-independent solution, namely, the solution which satisfies the spectral problem corresponding to the $L(\lambda=0)$-operator.

To circumvent singular behaviour of $\Psi_{t}(x,t,\lambda)$ at $\lambda=0$, we note from equation \eqref{y} that as $x\to \infty$ then $y\to \infty$ since $q>0$ for all $t$. Specifically we find $y\to\sqrt{u_0}x$ as $x\to\infty$, in which case
\[\Psi(X,t,0)=\Psi_{0}(y=X\sqrt{u_0},t,0)K_0\qquad \text{as }X\to\infty.\]
Hence, we conclude that $\Psi(X\to \infty,t,0) $ should be time-independent. Additionally, differentiating equation \eqref{sec3.2eq1} with respect to $t$, we find \[\Psi_t(X,t,0)=-\frac{\sigma_3}{2}\left(X_t-\frac{q_t}{2q}\right)\Psi(X,t,0),\]
and  $\Psi_t(X\to\infty,t,0)=-\frac{\sigma_3u_0}{2}\Psi(X\to\infty,t,0),$ having imposed $X_t=u(X,t)\to u_0$ as $X\to\infty$. Hence $\Psi(X,t,0)$ must be of the form
\[\Psi(X,t,0)=e^{-\frac{1}{2}\left(X-\ln\sqrt{q}-u_0 t\right )\sigma_3}\]
in order to have the appropriate asymptotic behaviour (the correction being independent of  $y$).
It follows that
\begin{equation}\label{sec3.2eq2}
e^{-\frac{1}{2}\left(X(y,t)-\ln\sqrt{q}-u_0 t \right)\sigma_3}
= g(y,t,0) e^{-\frac{1}{2\sqrt{u_0}}\sigma_3y}K_0,
\end{equation}
which gives a differential equation for $X$ since $\partial_y X=q^{-1/2}$, cf. equation \eqref{y}. Thus it provides the change of variables $x=X(y,t) $ in parametric form, where $y$ serves as the parameter. This of course is valid only in cases where the dressing factor is known and in what follows we shall explain how to construct it.

\subsection{The dressing factor with a real simple pole}
In the $SL(2)$ Zakharov-Shabat spectral problems, the simplest form of $g$ possesses one simple pole \cite{ZMNP,GVY}, which leads to the following:
\begin{proposition}\label{sec3.1prop1}
        The dressing factor $g(y,t,\lambda)$ is assumed to be of the form
        \begin{equation}\label{sec3.1eq4}
        g = \id+\frac{2 \lambda_1 B(y,t)}{\lambda-\lambda_1},\quad \text{where }\lambda_1 \in\RR
        \end{equation}
        and $B$ is a matrix-valued residue of rank 1.
\end{proposition}
By virtue of equation \eqref{gsym} and Proposition \ref{sec3.1prop1},  we deduce that the dressing factor must satisfy
\begin{equation}\label{sec3.1eq5}
\begin{aligned}
\left(\id + \frac{2\lambda_1 B}{\lambda-\lambda_1}\right)\left(\id-\frac{2 \lambda_1 J B^TJ}{\lambda+\lambda_1}\right) =\id,
\end{aligned}
\end{equation}
and taking residues as $\lambda\to\pm \lambda_1 $ we observe
\begin{equation}\label{sec3.1eq6}
\left\{
\begin{aligned}
&B\left(\id -JB^{T}J\right)=0\\
&\left(\id-B\right)JB^{T}J=0.
\end{aligned}
\right.
\end{equation}
Rewriting the matrix $B$ as
\begin{equation}\label{sec3.1eq7}
    B= \ket{n}\bra{m},\quad\text{with } \ket{n}=\left(\begin{matrix}n_1\\n_2\end{matrix}\right)\text{ and } \bra{m}=\left(\begin{matrix}m_1&m_2\end{matrix}\right),
\end{equation}
equations \eqref{sec3.1eq6}--\eqref{sec3.1eq7} combined with the symmetry relation \eqref{gsym} now ensure
\begin{equation}\label{sec3.1eq8}
    \bra{n}=\frac{\bra{m}J}{\bra{m}J\ket{m}}\Rightarrow B=\frac{\ket{m}\bra{m}J}{\bra{m}J\ket{m}},
\end{equation}
in which case $B=B^2$ meaning $B$ is a projector.
Moreover, equation \eqref{sec3.1eq8} combined with the symmetry relation $\bar{g}(y,t,\bar{\lambda}) = {g}(y,t,\lambda)$  of \eqref{gsym} also yields that $B$ should be real.

Replacing equation \eqref{sec3.1eq4} in equation \eqref{sec3.1eq3} and taking residues as $\lambda\to \lambda_1 $ and $\lambda\to\infty$, we have
\begin{equation}\label{sec3.1eq10}
\begin{split}
(h&-h_0)\sigma_3=2\lambda_1 [J,B], \\
  B_y& + \left(h_0\sigma_3+\lambda_1 J\right)B-B\left(h_0\sigma_3-\lambda_1 J\right)-2\lambda_1 BJB=0.
\end{split}
\end{equation}
Replacing equation \eqref{sec3.1eq8} in equation \eqref{sec3.1eq10}, multiplying everywhere by $J$ from the right and using $J\sigma_3J=-\sigma_3$ we have
\begin{multline}\label{sec3.1eq12}
\left(\ket{m_y}+\left(h_0\sigma_3+ \lambda_1 J\right)\ket{m}\right) \frac{\bra{m}}{\bra{m}J\ket{m}}+\frac{\ket{m}}{\bra{m}J\ket{m}}\left(\bra{m_y}+\bra{m}\left(h_0\sigma_3+\lambda_1 J\right)\right)\\
-\frac{2\bra{m_y}J\ket{m}\ket{m}\bra{m}}{\bra{m}J\ket{m}^2}-\frac{2\lambda_1\bra{m}J^2\ket{m}\ket{m}\bra{m}}{\bra{m}J\ket{m}^2}=0.
\end{multline}
Assuming
\begin{equation}\label{sec3.1eq13}
 \ket{m_y}+\left(h_0\sigma_3+\lambda_1 J\right) \ket{m} = 0,
\end{equation}
we also observe that  equation \eqref{sec3.1eq12} is satisfied identically provided \eqref{sec3.1eq13} holds. Thus $\ket{m}$ is an eigenvector of the  bare spectral problem (evaluated at $\lambda=-\lambda_1$), in which case $\ket{m}$ is known explicitly.

\section{Cuspons, solitons and the cuspon-soliton interaction}\label{sec4}
\subsection{The one-cuspon solution}\label{sec4.1}
Equation \eqref{sec3.1eq13} suggest that $\ket{m}$ satisfies the bare spectral problem \eqref{sec2.3eq1} with $\lambda=-\lambda_1$, i.e. with spectral operator \[L_{0}(y,t,-\lambda_1)=h_0\sigma_3+\lambda_1 J\]
Furthermore equation \eqref{sec3.1eq8} allows us to solve for $\ket{n}$ explicitly, thereby providing an explicit formula for the dressing factor $g(y,t,\lambda)$.
We can write
\begin{equation}\label{sec3.3eq1}
  \ket{m}=\Psi(y,t,-\lambda_1)\ket{m_0}
\end{equation}
where $\ket{m_0}$ is a constant vector, and $\Psi(y,t,-\lambda_1)\in SL(2)$ satisfies the bare spectral problem
\begin{equation}\label{sec3.3eq2}
\begin{cases}
 &\Psi_{y} +L_{0}(y,t,-\lambda_1)\Psi=0\\
 &\Psi_{t}-\frac{1}{2h_0}\left(u_0-\frac{1}{2\lambda_1^2}\right ) L_{0}(y,t,-\lambda_1)\Psi=0.
\end{cases}
\end{equation}
With $\lambda=-\lambda_1,$  we have $\Lambda = \sqrt{h_0^2+ \lambda_1^2} > h_0,$
\begin{equation}\label{sec3.2eq4}
 \sin(\theta)=\sqrt{\frac{\Lambda-h_0}{2\Lambda}}\qquad \cos(\theta)= \sqrt{\frac{h_0+\Lambda}{2\Lambda}},
\end{equation} thus we conclude that  $\theta$ is real.

It follows from equation \eqref{L-V} that
\begin{equation}\label{sec3.2eq5}
\left\{
\begin{aligned}
 &\Psi(y,t,-\lambda_1)=V e^{-\sigma_3\Omega(y,t)}V^{-1}\\
 &\Omega(y,t)=\Omega(y,t,\lambda_1)=\Lambda\left(y-\frac{1}{2h_0}\left(u_0-\frac{1}{2\lambda_1^2}\right )t\right)
\end{aligned}
\right.
\end{equation}
while equation \eqref{sec3.3eq1} now ensures
\begin{equation}\label{sec3.2eq6}
 \ket{m}=\left(\begin{matrix}\mu_{01}e^{-\Omega(y,t)}\sqrt{\frac{h_0+\Lambda}{2\Lambda}}-\mu_{02}e^{\Omega(y,t)}\sqrt{\frac{\Lambda-h_0}{2\Lambda}}\\ \mu_{01}e^{-\Omega(y,t)}\sqrt{\frac{\Lambda-h_0}{2\Lambda}}+\mu_{02}e^{\Omega(y,t)}\sqrt{\frac{h_0+\Lambda}{2\Lambda}}\end{matrix}\right)
\end{equation}
where the real coefficients $\mu_{01}$ and $\mu_{02}$ are given by
\begin{equation}\label{sec3.2eq7}
  V^{-1}\ket{m_0}=\left(\begin{matrix}\mu_{01}\\ \mu_{02}\end{matrix}\right).
\end{equation}
Making the replacement $\nu_1=\frac{\mu_{01}}{\sqrt{2\Lambda}}$ and $\nu_2=\frac{\mu_{02}}{\sqrt{2\Lambda}}$ we simplify $\ket{m}=(m_1, m_2)^T$ according to
\begin{equation}\label{sec3.2eq8}
  \ket{m}=\left(\begin{matrix}\nu_1e^{-\Omega(y,t)}\sqrt{h_0+\Lambda}-\nu_2e^{\Omega(y,t)}\sqrt{\Lambda-h_0}\\ \nu_1e^{-\Omega(y,t)}\sqrt{\Lambda-h_0}+\nu_2e^{\Omega(y,t)}\sqrt{h_0+\Lambda}\end{matrix}\right).
\end{equation}
It follows from equations \eqref{sec3.1eq4}, \eqref{sec3.1eq8} and \eqref{sec3.2eq8} that
\begin{equation}\label{sec3.2eq9}
g(y,t;0)= \left(\begin{array}{cc} 0 &  -\frac{m_1}{m_2} \\ -\frac{m_2}{m_1} & 0 \end{array}\right),
\end{equation}
while letting $y\to\infty$ we also deduce the form of $K_0$, namely
\begin{equation}\label{sec3.2eq10}
K_0=g^{-1}(y\to -\infty ,t;0)= \left(\begin{array}{cc} 0 &  -\sqrt{\frac{\Lambda + h_0}{\Lambda - h_0} }\\ -\sqrt{\frac{\Lambda - h_0}{\Lambda + h_0} } & 0 \end{array}\right).
\end{equation}
Equations \eqref{sec3.2eq2} and \eqref{sec3.2eq8}--\eqref{sec3.2eq10} now yield the following differential equation
\begin{equation}\label{sec4.1eq6}
\begin{split}
 \frac{\partial X}{\partial y}e^{X-u_0t}&=e^{-\frac{y}{\sqrt{u_0}}}\frac{m_2^2(\Lambda +h_0)}{m_1^2 (\Lambda-h_0)}=e^{-\frac{y}{\sqrt{u_0}}}\left[\frac{\nu_1 |\lambda_1|e^{-\Omega(y,t)} +\nu_2(\Lambda+h_0)e^{\Omega(y,t)}}{\nu_1 |\lambda_1| e^{-\Omega(y,t)}-\nu_2 (\Lambda-h_0)e^{\Omega(y,t)}}\right]^2.
\end{split}
\end{equation}

Implementing the change of variables $(y,t)\to(-y,-t)$ and  observing that $\Omega(-y,-t)=-\Omega(y,t)$,  the differential equation for $X$ becomes
\begin{equation}\label{sec4.1eq6a}
\begin{split}
\frac{dX}{dy}e^{X-u_0t -\frac{y}{\sqrt{u_0}}}&=\left[\frac{s e^{\Omega(y,t)} +(\Lambda+h_0)e^{-\Omega(y,t)}}{s e^{\Omega(y,t)} -(\Lambda-h_0)e^{-\Omega(y,t)}}\right]^2,\qquad s:=\frac{\nu_1|\lambda_1|}{\nu_2}
\end{split}
\end{equation}
The reason for doing so is the following: The Camassa-Holm equation written in terms of the  $(y,t)$-variables yields the so-called ACH equation (see for instance \cite{RS,H99,PLA}), which is invariant under $(y,t)\to(-y,-t)$. Thus choosing  any solution of the ACH equation and imposing the change of variables $(y,t)\to(-y,-t)$, we obtain another solution once we determine $x=X(y,t)$. In other words, if $X(y,t)$ is a solution of the Camassa-Holm equation then so too is $x=\tilde{X}(y,t)=X(-y,-t)$. Moreover, with this change of variables we also have $x\to \frac{y}{\sqrt{u_0}}$ as $y\to \pm \infty$, cf. equation \eqref{y}.

Explicitly this change of variables imposes the following transformation on our differential equation for $X$
\begin{equation}\label{yto-y}
\begin{split}
 \frac{d\tilde{X}(-y,-t)}{d(-y)}e^{\tilde{X}(-y,-t)-u_0(-t)}&=e^{-\frac{y}{\sqrt{u_0}}}\left[\frac{\nu_1 |\lambda_1|e^{-\Omega(y,t)} +\nu_2(\Lambda+h_0)e^{\Omega(y,t)}}{\nu_1 |\lambda_1| e^{-\Omega(y,t)}-\nu_2 (\Lambda-h_0)e^{\Omega(y,t)}}\right]^2,\\
 \frac{dX(y,t)}{dy}e^{X(y,t)-u_0t}&=e^{\frac{y}{\sqrt{u_0}}}\left[\frac{\nu_1 |\lambda_1|e^{-\Omega(-y,-t)} +\nu_2(\Lambda+h_0)e^{\Omega(-y,-t)}}{\nu_1 |\lambda_1| e^{-\Omega(-y,-t)}-\nu_2 (\Lambda-h_0)e^{\Omega(-y,-t)}}\right]^2.
\end{split}
\end{equation}
Formally this may be integrated by separation of variables, however we may also look for a solution in the form
\begin{equation}\label{sec4.1eq7}
 X(y,t)=\frac{y}{\sqrt{u_0}}+u_0t+\ln\left \vert\frac{\mathcal{A}_{C}}{\mathcal{B}_{C}}\right \vert
\end{equation}
with
\begin{equation}\label{sec4.1eq8}
 \mathcal{A}_{C}=a_1e^{\Omega(y,t)}+a_2e^{-\Omega(y,t)}\qquad \mathcal{B}_{C}=s e^{\Omega(y,t)}- (\Lambda-h_0)e^{-\Omega(y,t)}.
\end{equation}
Replacing equations \eqref{sec4.1eq7}--\eqref{sec4.1eq8} in equation \eqref{sec4.1eq6}, we conclude that $a_1=s$, $a_2=-(h_0+\Lambda^{2})/(\Lambda-h_0)$ and thus
\begin{equation}
X(y,t)=\frac{y}{\sqrt{u_0}}+u_0 t+\ln\left\vert\frac{(h_0+\Lambda)^2e^{-\Omega(y,t)}
+\gamma e^{\Omega(y,t)}}{(h_0-\Lambda)^2e^{-\Omega(y,t)} +\gamma e^{\Omega(y,t)}}\right\vert , \qquad \gamma =-s(\Lambda-h_0).
\end{equation}
When $\gamma >0$ this expression can also be written as
\begin{equation}
X(y,t)=\frac{y}{\sqrt{u_0}}+u_0t+\ln\left\vert\frac{\frac{\gamma+(h_0+\Lambda)^2}{\gamma -(h_0+\Lambda)^2}\coth {\Omega(y,t)}+1}{\frac{\gamma +(h_0-\Lambda)^2,}{\gamma -(h_0-\Lambda)^2} \coth{\Omega(y,t)} +1}\right\vert + \text{const}.
\end{equation}
We introduce the constant
\[U=\frac{\Lambda}{h_0} >1\]
and  choose constants $\{\nu_1,\nu_2\}$ such that $\gamma=\lambda_1^2$, thereby simplifying the expression for $X(y,t)$, which is now given by
\begin{equation}
X(y,t)=\frac{y}{\sqrt{u_0}}+u_0t+\ln\left\vert\frac{U\coth {\Omega(y,t)}-1}
{U \coth{\Omega(y,t)} +1}\right\vert.
\end{equation}
having ignored a trivial additive constant.

We recall that $\lambda^2(k)=-\frac{1}{u_0}\left(k^2+\frac{1}{4}\right)$, while the discrete spectrum corresponds to $k=i\kappa_n$ with $\kappa_{n}\in(\frac{1}{2}, \infty),$ thus ensuring the discrete spectral parameter $\lambda_n$ is real.
Written in terms of the parameter $U=\frac{\Lambda(\lambda_1)}{h_0}>0$, the one-cuspon solution is now given by $u(X(y,t),t)=X_{t}(y,t)$ with
\begin{equation}\label{eq4.18}
\begin{cases}
X(y,t)=\frac{y}{\sqrt{u_0}}+u_0t+\ln\left\vert\frac{U\coth\Omega(y,t)-1}{U\coth\Omega(y,t)+1}\right\vert\\
u(X(y,t),t)=u_0+\frac{U^2\left(u_0-\frac{1}{2\lambda_1^2}\right)(1-\coth^2\Omega(y,t))}{1-U^2\coth^2\Omega(y,t)}\\
\Omega(y,t,\lambda_1)=\Lambda(\lambda_1)\left(y-\frac{1}{2h_0}\left(u_0-\frac{1}{2\lambda_1^2}\right)t\right)\\
\Lambda(\lambda_1)=\sqrt{h_0^2+\lambda_1^2}.
\end{cases}
\end{equation}
We can rewrite the second line of equation \eqref{eq4.18} as
$$ u-u_0=\left(u_0-\frac{1}{2\lambda_1^2}\right)\frac{U^2(1-\tanh^2 \Omega)}{U^2 - \tanh^2 \Omega} $$
and we note that $u>u_0$ when $u_0>\frac{1}{2\lambda_1^2}$ with the corresponding profile called an {\it a cuspon},  left-panel Figure \ref{fig0}. The case $u<u_0$ when $u_0<\frac{1}{2\lambda_1^2}$ is called an {\it anti-cuspon},  right-panel Figure \ref{fig0}.
\begin{figure}[h!]
\includegraphics[width=0.5\textwidth]{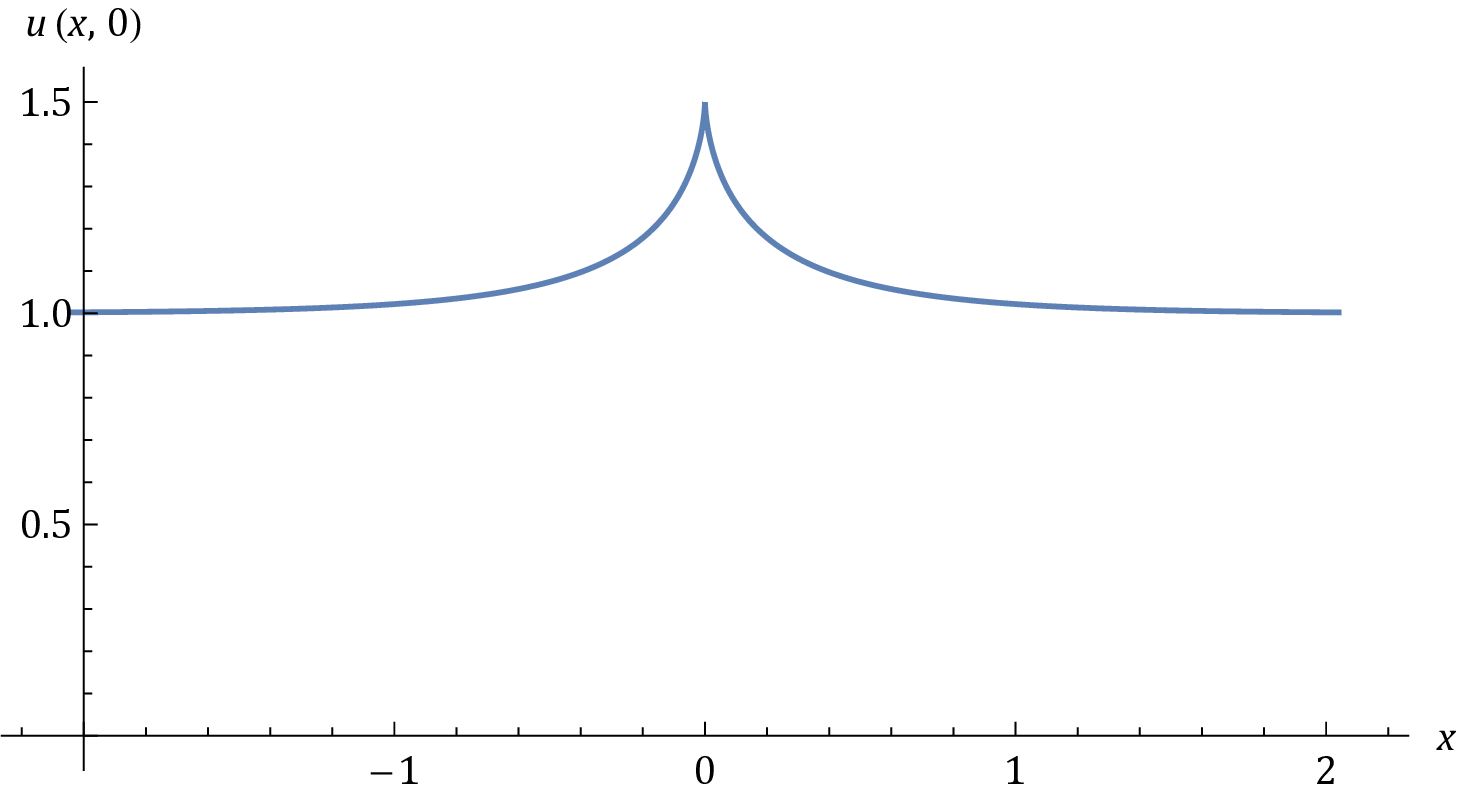}%
\includegraphics[width=0.5\textwidth]{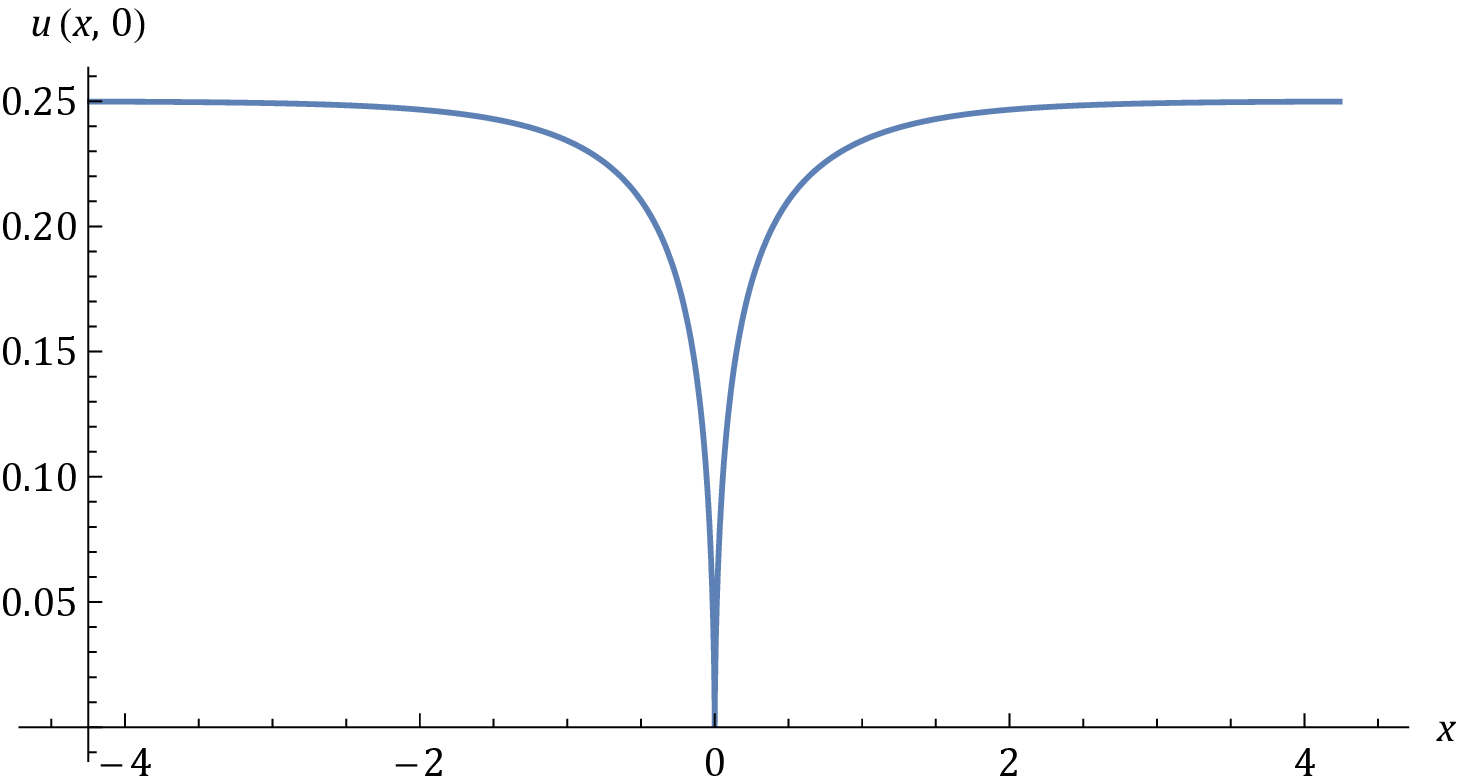}%
\caption{One-cuspon and one-anticuspon solution profiles of the Camassa-Holm equation, $\lambda_1=1.0.$ On the left panel $u_0=1.0,$ on the right panel $u_0=0.25.$ }\label{fig0}
\end{figure}

Let us now evaluate the slope $u_X$ of the cuspon profile to investigate the discontinuity at the cusp. We have
$$ X_y = \frac{(U^2-1)\sinh ^2 \Omega}{\sqrt{u_0}[U^2 + (U^2-1)\sinh^2 \Omega ]}$$ along with $\Lambda=Uh_0$ and $2h_0 \sqrt{u_0}=1$ and so we find
$$u_X(X(y,t),t)= \frac{1}{X_y} \frac{\partial u}{\partial y}=-\frac{U^3 \left( u_0 - \frac{1}{2\lambda_1^2}\right)}{U^2 + (U^2-1)\sinh^2 \Omega} \coth \Omega.$$
Thus, in the cuspon case ($u_0 > \frac{1}{2\lambda_1^2}$) we have $u_X \to  \infty$ when $\Omega \to 0^{-}$ left of the cusp at $\Omega=0$, and $u_X \to  -\infty$ when $\Omega \to 0^{+}$ right of the cusp. For the anti-cuspon the signs change in an obvious way.
The cusp is located at $\Omega=\Lambda(\lambda_1)\left(y-\frac{1}{2h_0}\left(u_0-\frac{1}{2\lambda_1^2}\right)t\right)=0$ and moves with a constant velocity $$\frac{dy}{dt}= \frac{1}{2h_0}\left(u_0-\frac{1}{2\lambda_1^2}\right)  $$ with respect to the $y$-axis. This can be both positive (for the cuspons) and negative (for the anti-cuspons). However since $y$ is merely a parameter, the velocity should be measured with respect to the physical $x$-axis. Noting that when $\Omega=0,$ $X(y,t)=\frac{y}{\sqrt{u_0}}+u_0t+$const, we find that the cuspon (anti-cuspon) velocity is
$$\frac{dx}{dt}=2 \left(u_0-\frac{1}{4\lambda_1^2}\right)  $$ thereby indicating that the threshold velocity is $u_0=\frac{1}{4\lambda_1^2}$. Thus, the cuspon solution is always right-moving, since $u_0>\frac{1}{2\lambda_1^2}>\frac{1}{4\lambda_1^2},$   while the anti-cuspon solution is either right-moving, when $\frac{1}{2\lambda_1^2}>u_0>\frac{1}{4\lambda_1^2}$
or left-moving, when $u_0<\frac{1}{4\lambda_1^2}$. The special case $u_0=\frac{1}{4\lambda_1^2}$ therefore corresponds to a ``standing'' anti-cuspon.

\subsection{The one-soliton solution}\label{sec4.2}
The one soliton solution can be obtained in a similar way by a dressing factor with a simple imaginary pole $i\omega_1$:\begin{equation}\label{g1sol}
g(y,t,\lambda) = \id +\frac{2i\omega_1 A_1 (y,t)}{\lambda-i\omega_1}.
\end{equation}
The details can be found in \cite{ILO2017}, and the solution is
\begin{equation}
\begin{cases}
&X(y,t)=2h_0y+u_0t+\ln\abs{\frac{U\tanh\Omega(y,t)-1}{U\tanh\Omega(y,t)+1}}\\
&u(X,t)  = u_0+ \frac{U^2\left(u_0+\frac{1}{2\omega^2}\right)(1-\tanh^2 \Omega)}{1-U^2 \tanh^2 \Omega}\\
&\Omega(y,t)=\Lambda y-\frac{U}{2}\left(u_0+\frac{1}{2\omega^2}\right)t\\
&\Lambda=\sqrt{h_0^2-\omega^2},
\end{cases}
\end{equation}
where as in the previous case the solution is obtained via $X_{t}(y,t)=u(X(y,t),t)$.
We note that as $y\to\pm\infty$, then $\tanh(y) \to \pm 1$ and $u\to u_0$, which we observe in the soliton profile shown in Figure \ref{fig1}.
Interestingly, choosing constants such that $c(h_0-\Lambda)=-\gamma^2<0$, $X(y,t)$ is no longer a monotonic function for all $y\in \mathbb{R}$ meaning the solution $u(x,t)$ becomes a function with discontinuities,  and the reader is referred to \cite{RS} for further discussion where such solutions are termed {\it unphysical}.
\begin{figure}[h!]
 \centering
 \includegraphics[width=0.5 \textwidth]{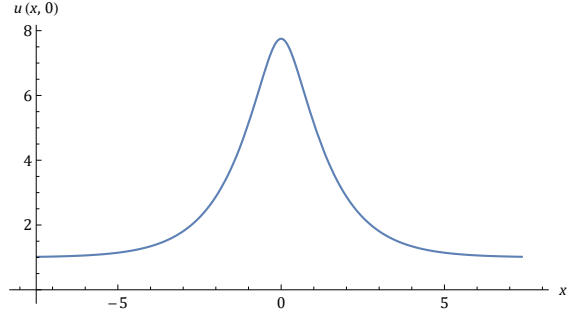}
 \caption{The one-soliton solution at $t=0$ where $u_0=1$ and $\omega=0.25$.}\label{fig1}
\end{figure}

\subsection{The two-cuspon solution}\label{sec4.3}
The dressing factor associated with the two-cuspon solution has two real simple poles  $\lambda_1$ and $\lambda_2$, with residues $2\lambda_k B_k$ ($k=1,2.$). Extending Proposition \ref{sec3.1prop1}, this dressing factor is of the form
\begin{equation}\label{sec4.2eq1}
g(y,t,\lambda) = \id+\frac{2\lambda_1 B_1 (y,t)}{\lambda-\lambda_1}+\frac{2\lambda_2 B_2 (y,t)}{\lambda-\lambda_2},
\end{equation}
where $B_1$ and $B_2$ are two matrix valued residues. The $\mathbb{Z}_2$ reduction $  \bar{g}(y,t,\bar{\lambda})= g(y,t,\lambda)  $ necessitates $B_k$ to be real.

Applying equation \eqref{sec3.1eq3} to the dressing factor $g$ as given by equation \eqref{sec4.2eq1} ensures the corresponding equations for the residues, namely
\begin{equation}\label{sec4.2eq3}
B_{k,y}+h \sigma_3 B_k - B_k h_0\sigma_3 - \lambda_k [J, B_k]=0 \quad \text{ for }k=1,2.
\end{equation}
The matrix valued residues $B_k$ are of the form
\begin{equation}\label{sec4.2eq4}
B_1 = \ket{n}\bra{m}, \qquad  B_2 = \ket{N}\bra{M},
\end{equation}
where the vectors $\ket{n}$, $\ket{N}$, $\bra{m}$ and $\bra{M}$ are found to satisfy
\begin{equation}\label{sec4.2eq5}
\begin{aligned}
&\partial_y \ket{n} + (h\sigma_3-\lambda_1 J)\ket{n}=0, \qquad \partial_y\bra{m}= \bra{m} (h_0\sigma_3-\lambda_1 J)\\
&\partial_y \ket{N} + (h\sigma_3-\lambda_2 J)\ket{N}=0, \qquad \partial_y\bra{M}= \bra{M} (h_0\sigma_3-\lambda_2 J).
\end{aligned}
\end{equation}
The vectors  $\bra{m},\bra{M}$ satisfy the bare equations and therefore are known in principle, while the reality condition can be satisfied by assuming $\ket{m}$, $\ket{M}$, $\ket{n}$, and $\ket{N}$ are all real. The reduction given in equation \eqref{gsym} leads to
\begin{equation}\label{sec4.2eq6}
 \left[\id + \frac{2\lambda_1 B_1}{\lambda-\lambda_1}+\frac{2\lambda_2 B_2}{\lambda- \lambda_2}\right]\left[\id - \frac{2\lambda_1J B_1^{T}J}{\lambda+ \lambda_1}-\frac{2\lambda_2 J B_2^{T}J}{\lambda+\lambda_2}\right]=\id,
\end{equation}
which is identically satisfied for all $\lambda.$ Thus, the residues obtained at $\lambda_1$ and $\lambda_2,$ ensure
\begin{equation}\label{sec4.2eq7}
\begin{split}
 B_1 \left(\id - J B_1^{T}J - \eta_2 J B_2^{T}J\right)& =0,\qquad \eta_2=\frac{2\lambda_2}{\lambda_1+\lambda_2},  \\  B_2 \left(\id -\eta_1 J B_1^{T}J -  J B_2^{T}J\right)& =0,\qquad \eta_1=\frac{2\lambda_1}{\lambda_1+\lambda_2}.
\end{split}
\end{equation}
Using equations \eqref{sec4.2eq4}-\eqref{sec4.2eq7} we obtain the following system relating the bare and dressed eigenvectors:
\begin{equation}\label{sec4.2eq8}
\begin{aligned}
 \bra{m}&=\bra{m} J \ket{m} \bra{n} J+ \eta_2 \bra{m} J \ket{M} \bra{N} J,\\
 \bra{M}&=\eta_1 \bra{M} J \ket{m} \bra{n} J+  \bra{M} J \ket{M} \bra{N}J.
\end{aligned}
\end{equation}
Hence, the dressed vectors $\ket{n}$ and $\ket{N}$ may be written explicitly in terms of the bare vectors $\ket{m}$ and $\ket{M}$ as follows
\begin{equation}\label{sec4.2eq9}
\begin{aligned}
 \ket{n}&=\frac{1}{\Delta}\big(\bra{M} J \ket{M} J\ket{m} - \eta_2 \bra{m} J \ket{M} J\ket{M} \big),\\
 \ket{N}&=\frac{1}{\Delta}\big(\bra{m} J \ket{m} J\ket{M} - \eta_1 \bra{M} J \ket{m} J\ket{m} \big),
\end{aligned}
\end{equation}
with
\begin{equation}\label{sec4.2eq10}
\begin{aligned}
\Delta&=\bra{M} J \ket{M}\bra{m} J \ket{m}-\eta_1 \eta_2\bra{m} J \ket{M}^2\\
      &=(\eta_1 m_2 M_1-\eta_2 m_1 M_2)(\eta_1 m_1 M_2 - \eta_2 m_2 M_1).
\end{aligned}
\end{equation}
Thus, the residues $B_k$ can be expressed in terms of components of the known vectors
\begin{equation}\label{mn}
\bra{m}=\bra{m_{(0)}}\Psi_0^{-1}(y,t,\lambda_1),\qquad  \bra{M}=\bra{M_{(0)}}\Psi_0^{-1}(y,t,\lambda_2),
\end{equation}
where $\bra{m_{(0)}},\bra{M_{(0)}} $ are arbitrary constant vectors.

The dressing factor  $g(\lambda)\in SL(2)$ (cf. equation \eqref{sec4.2eq1}), when evaluated at $\lambda=0$, is given by
\begin{equation}\label{sec4.2eq12}
\begin{aligned}
g(y,t;0)&=\id - 2(B_1+B_2)\\
        &=\mathrm{diag}(g_{11},g_{22})\\
        &=\mathrm{diag}\left(\frac{\lambda_1 M_1 m_2 -\lambda_2 M_2 m_1}{\lambda_1 M_2 m_1-\lambda_2 M_1 m_2},\frac{\lambda_1 M_2 m_1-\lambda_2 M_1 m_2}{\lambda_1 M_1 m_2 -\lambda_2 M_2 m_1}  \right),
\end{aligned}
\end{equation}
while the differential equation for $X(y,t)$ is of the form
\begin{equation}\label{sec4.2eq13}
(\partial_y X) e^{X-2h_0 y-u_0 t} =g_{22}^2=\left(\frac{\lambda_1 M_2 m_1-\lambda_2 M_1 m_2}{\lambda_1 M_1 m_2 -\lambda_2 M_2 m_1}\right)^2.
\end{equation}
In the case of the two-cuspon solution, equations \eqref{Psi0}-\eqref{L-V} now become
\begin{equation}
 \Psi(y,t,\lambda_k)=V_{k}e^{-\sigma_3\Omega_k}V_k^{-1}
\end{equation}
with
\begin{equation}
\left\{
\begin{aligned}
 &\Omega_k(y,t)=\Lambda_{k}\left(y-\frac{1}{2h_0}\left(u_0-\frac{1}{2\lambda_k^2}\right )t\right),\quad
 \Lambda_k=\sqrt{h_0^2+\lambda_k^2}>h_0,  \\
 &V_{k}=\left(\begin{matrix}\cos(\theta_k)&-\sin(\theta_k)\\ \sin(\theta_k)&\cos(\theta_k)\end{matrix}\right),\\
 &\cos(\theta_k)=\sqrt{\frac{\Lambda_k+h_0}{2\Lambda_k}}\qquad \sin(\theta_k)=\sqrt{\frac{\Lambda_k-h_0}{2\Lambda_k}}
 \end{aligned}
 \right.
\end{equation}
for $k=1, 2$. Using equation \eqref{mn} and noticing that $\bra{m_{(0)}}V_1= (\mu_1, \mu_2)$ is a constant vector, upon choosing $\mu_1, \mu_2$ to be real and positive then explicitly we have
\begin{equation*}
\begin{split}
  m_1&=\sqrt{h_0+\Lambda_1}\sqrt{\frac{\mu_1}{\mu_2}}e^{\Omega_1(y,t)}-\sqrt{\Lambda_1-h_0}\sqrt{\frac{\mu_2}{\mu_1}}e^{-\Omega_1(y,t)},\\ m_2&=\sqrt{\Lambda_1-h_0}\sqrt{\frac{\mu_1}{\mu_2}}e^{\Omega_1(y,t)}+\sqrt{\Lambda_1+h_0}\sqrt{\frac{\mu_2}{\mu_1}}
e^{-\Omega_1(y,t)} ,
 \end{split}
\end{equation*}
up to an overall constant foactor $\sqrt{\mu_1 \mu_2}(2\Lambda_1)^{-1/2}$, which we may neglect since it ultimately cancels, cf. equation \eqref{g2cuspon}.

We redefine $\Omega_1(y,t)$ by an additive constant, which becomes
\begin{equation}
    \Omega_1(y,t)=\Lambda_{1}\left(y-\frac{t}{2h_0}\left(u_0-\frac{1}{2\lambda_1^2}\right)\right)+ \ln \sqrt{\frac{\mu_1}{\mu_2}},
\end{equation}
whereupon the vector components $m_1$ and $m_2$ become
\begin{equation} \label{m12}
\begin{split}
  m_1&=\sqrt{\Lambda_1+h_0}e^{\Omega_1(y,t)}-\sqrt{\Lambda_1-h_0}e^{-\Omega_1(y,t)}\\
  m_2&=\sqrt{\Lambda_1-h_0}e^{\Omega_1(y,t)}+\sqrt{\Lambda_1+h_0}e^{-\Omega_1(y,t)}.
\end{split}
\end{equation}
Similarly, taking the constant vector $\bra{M_{(0)}}V_2= (\nu_1, - \nu_2)$ with $\nu_1, \nu_2$ real and positive, we have
\begin{equation} \label{M12}
\begin{split}
  M_1&=\sqrt{\Lambda_2+h_0}e^{\Omega_2(y,t)}+\sqrt{\Lambda_2-h_0}e^{-\Omega_2(y,t)},\\ M_2&=\sqrt{\Lambda_2-h_0}e^{\Omega_2(y,t)}-\sqrt{\Lambda_2+h_0}e^{-\Omega_2(y,t)},\\
 \Omega_2(y,t)&=\Lambda_{2}\left(y-\frac{t}{2h_0}\left(u_0-\frac{1}{2\lambda_2^2}\right)\right)+ \ln \sqrt{\frac{\nu_1}{\nu_2}}.
 \end{split}
\end{equation}
The dressing factor component $g_{22}$ now takes the form
\begin{equation}\label{g2cuspon}
g_{22}=\frac{\lambda_1 M_2 m_1-\lambda_2 M_1 m_2}{\lambda_1 M_1 m_2  -\lambda_2 M_2 m_1}=\frac{\mathcal{T}_{CC}}{\mathcal{B}_{CC}}
\end{equation}
whose explicit form we deduce from equations \eqref{m12}--\eqref{M12}. We note in particular that the denominator of this expression is given by
\begin{equation}
\begin{split}
\mathcal{B}_{CC}=&\lambda_1 \lambda_2 \left(\frac{\Lambda_1 - \Lambda_2}{\sqrt{(\Lambda_1-h_0)(\Lambda_2-h_0)}}e^{\Omega_1+\Omega_2}+ \frac{\Lambda_1 - \Lambda_2}{\sqrt{(h_0+\Lambda_1)(h_0+\Lambda_2)}}e^{-\Omega_1-\Omega_2}  \right. \\
  & + \left.  \frac{\Lambda_1 + \Lambda_2}{\sqrt{(\Lambda_1-h_0)(\Lambda_2+h_0)}}e^{\Omega_1-\Omega_2} + \frac{\Lambda_1 + \Lambda_2}{\sqrt{(\Lambda_1+h_0)(\Lambda_2-h_0)}}e^{-\Omega_1+\Omega_2} \right)
\end{split}
\end{equation}
and introducing the constants
\begin{equation}
n_k=\sqrt[4]{\frac{\Lambda_k+h_0}{\Lambda_k-h_0}}\quad \text{ for }k = 1, 2,
\end{equation}
we obtain
\begin{equation}
\mathcal{B}_{CC}=\sqrt{\lambda_1 \lambda_2} (\Lambda_1^2-\Lambda_2^2) \left(n_1 n_2\frac{e^{\Omega_1+\Omega_2}}{\Lambda_1+\Lambda_2} +\frac{1}{n_1 n_2}\frac{e^{-\Omega_1-\Omega_2}}{\Lambda_1+\Lambda_2}
+\frac{n_1}{n_2}\frac{e^{\Omega_1-\Omega_2}}{\Lambda_1-\Lambda_2}+
\frac{n_2}{n_1}\frac{e^{-\Omega_1+\Omega_2}}{\Lambda_1-\Lambda_2} \right).
\end{equation}
Similarly it is found that
\begin{equation}
\mathcal{T}_{CC}=\sqrt{\lambda_1 \lambda_2} (\Lambda_1^2-\Lambda_2^2) \left(\frac{1}{n_1 n_2}\frac{e^{\Omega_1+\Omega_2}}{\Lambda_1+\Lambda_2} +n_1 n_2\frac{e^{-\Omega_1-\Omega_2}}{\Lambda_1+\Lambda_2}
-\frac{n_2}{n_1}\frac{e^{\Omega_1-\Omega_2}}{\Lambda_1-\Lambda_2}-
\frac{n_1}{n_2}\frac{e^{-\Omega_1+\Omega_2}}{\Lambda_1-\Lambda_2} \right).
\end{equation}
We seek a solution of equation  \eqref{sec4.2eq13} in the form of equation \eqref{sec4.1eq7}

$$ X(y,t)=\frac{y}{\sqrt{u_0}}+u_0t+\ln \frac{\mathcal{A}_{CC}}{\mathcal{B}_{CC}}  $$
with
\begin{equation}
\mathcal{A}_{CC}=\alpha_1 e^{\Omega_1+\Omega_2}+\alpha_2 e^{-\Omega_1-\Omega_2}+\alpha_3 e^{\Omega_1-\Omega_2}+\alpha_4 e^{-\Omega_1+\Omega_2},
\end{equation}
for some as yet unknown constants $\left\{\alpha_l\right\}_{l=1}^{4}$. Equation \eqref{sec4.2eq13}  ensures that $\mathcal{A}_{CC}$ must satisfy the following
\begin{equation}
2h_0\mathcal{A}_{CC}\mathcal{B}_{CC}+\mathcal{B}_{CC}\partial_y\mathcal{A}_{CC}-\mathcal{A}_{CC}\partial_{y}\mathcal{B}_{CC}=2h_0\mathcal{T}_{CC}^2
\end{equation}
whose solution is given by
\begin{equation}
\mathcal{A}_{CC}=\sqrt{\lambda_1 \lambda_2} (\Lambda_1^2-\Lambda_2^2) \left(\frac{1}{n_1^3 n_2^3}\frac{e^{\Omega_1+\Omega_2}}{\Lambda_1+\Lambda_2}+n_1^3 n_2^3\frac{e^{-\Omega_1-\Omega_2}}{\Lambda_1+\Lambda_2}
+\frac{n_2^3}{n_1^3}\frac{e^{\Omega_1-\Omega_2}}{\Lambda_1-\Lambda_2}+
\frac{n_1^3}{n_2^3}\frac{e^{-\Omega_1+\Omega_2}}{\Lambda_1-\Lambda_2} \right).
\end{equation}
The ratio $\mathcal{A}_{CC}/\mathcal{B}_{CC}$ may be written as
\begin{equation}
\frac{\mathcal{A}_{CC}}{\mathcal{B}_{CC}}=n_1^4 n_2^4\frac{1+ \frac{1}{n_1^6 }\frac{\Lambda_1+\Lambda_2}{\Lambda_1-\Lambda_2}e^{2\Omega_1}+ \frac{1}{n_2^6 }\frac{\Lambda_1+\Lambda_2}{\Lambda_1-\Lambda_2}e^{2\Omega_2} +\frac{e^{2\Omega_1+2\Omega_2}}{n_1^6n_2^6} }{1+ n_1^2\frac{\Lambda_1+\Lambda_2}{\Lambda_1-\Lambda_2}e^{2\Omega_1}+ n_2^2 \frac{\Lambda_1+\Lambda_2}{\Lambda_1-\Lambda_2}e^{2\Omega_2}+ n_1^2n_2^2e^{2\Omega_1+2\Omega_2} }
\end{equation}
which we simplify by means of the following re-definitions:
\begin{equation} \label{Omega12}
\begin{split}
 \Omega_1(y,t)&=\Lambda_{1}\left(y-\frac{t}{2h_0}\left(u_0-\frac{1}{2\lambda_1^2}\right)\right)+ \ln \sqrt{\frac{\mu_1}{\mu_2}}-\ln n_1 + \frac{1}{2}\ln \frac{\Lambda_1+\Lambda_2}{\Lambda_1 - \Lambda_2},\\
 \Omega_2(y,t)&=\Lambda_{2}\left(y-\frac{t}{2h_0}\left(u_0-\frac{1}{2\lambda_2^2}\right)\right)+ \ln \sqrt{\frac{\nu_1}{\nu_2}}-\ln n_2 + \frac{1}{2}\ln \frac{\Lambda_1+\Lambda_2}{\Lambda_1 - \Lambda_2}.
 \end{split}
\end{equation}
Alternatively, these may be simply written as
\begin{equation} \label{Omegak}
 \Omega_k(y,t)=\Lambda_{k}\left(y-\frac{t}{2h_0}\left(u_0-\frac{1}{2\lambda_k^2}\right)\right)+\xi_k\quad\text{ for }k=1, 2,
\end{equation}
where the constants $\xi_k$ relate to the initial sepparation of the solitons.  It follows that

 $$ \frac{\mathcal{A}_{CC}}{\mathcal{B}_{CC}}=n_1^4 n_2^4\frac{1+ \frac{1}{n_1^4 }e^{2\Omega_1}+ \frac{1}{n_2^4 }e^{2\Omega_2} +\left(\frac{\Lambda_1-\Lambda_2}{\Lambda_1+\Lambda_2}\right)^2\frac{e^{2\Omega_1+2\Omega_2}}{n_1^4n_2^4} }{1+ n_1^4e^{2\Omega_1}+ n_2^4e^{2\Omega_2} + \left(\frac{\Lambda_1-\Lambda_2}{\Lambda_1+\Lambda_2}\right)^2n_1^4n_2^4e^{2\Omega_1+2\Omega_2} }  $$ and thus
\begin{equation} \label{solX}
\begin{cases}
&X(y,t)=\frac{y}{\sqrt{u_0}}+u_0t+\ln \left(\frac{1+ \frac{1}{n_1^4 }e^{2\Omega_1}+ \frac{1}{n_2^4 }e^{2\Omega_2} +\left(\frac{\Lambda_1-\Lambda_2}{\Lambda_1+\Lambda_2}\right)^2\frac{e^{2\Omega_1+2\Omega_2}}{n_1^4n_2^4} }{1+ n_1^4e^{2\Omega_1}+ n_2^4e^{2\Omega_2} + \left(\frac{\Lambda_1-\Lambda_2}{\Lambda_1+\Lambda_2}\right)^2n_1^4n_2^4e^{2\Omega_1+2\Omega_2} } \right)  +X_0\\
&u(X(y,t),t)=\frac{\partial X}{\partial t},
\end{cases}
\end{equation}
which yields a solution of a form similar to that found in \cite{M1}.  Here $X_0$ is an overall additive constant that appears due to the translational invariance of the problem.

Given the expression \eqref{Omegak}, we see that the phase velocity $\Lambda_k\left(u_0-\frac{1}{2\lambda_k^2}\right)$ may be both positive and negative, thereby ensuring cuspon solutions may be left-moving or right-moving. In Figure \ref{fig3-1a} we show snapshots of a two cuspon solution with both cuspons right moving.
Such a solution arises when $u_{0}>\frac{1}{2\lambda_{\min}^2}$ where $\lambda_{\min}=\min\left\{\lambda_{1},\lambda_2\right\}$.
Conversely, we may have a two-cuspon solution with both cuspons left-moving when $u_{0}<\frac{1}{2\lambda_{\max}^2}$ with $\lambda_{\max}=\max\left\{\lambda_{1},\lambda_2\right\}$, and an example of such a solution is shown in Figure \ref{fig3-1b}. The third category of two-cuspon solution is a mixture of both, that is to say, when one cuspon is right moving while the other cuspon is left moving. Such a solution occurs when the asymptotic value $u_0$ is constrained by
$\frac{1}{2\lambda_{\max}}<u_0<\frac{1}{2\lambda_{\min}}$ and an example of such a solution is illustrated in Figure \ref{fig3-1c}.
\begin{figure}[h!]
\centering
\includegraphics[width=0.7\textwidth]{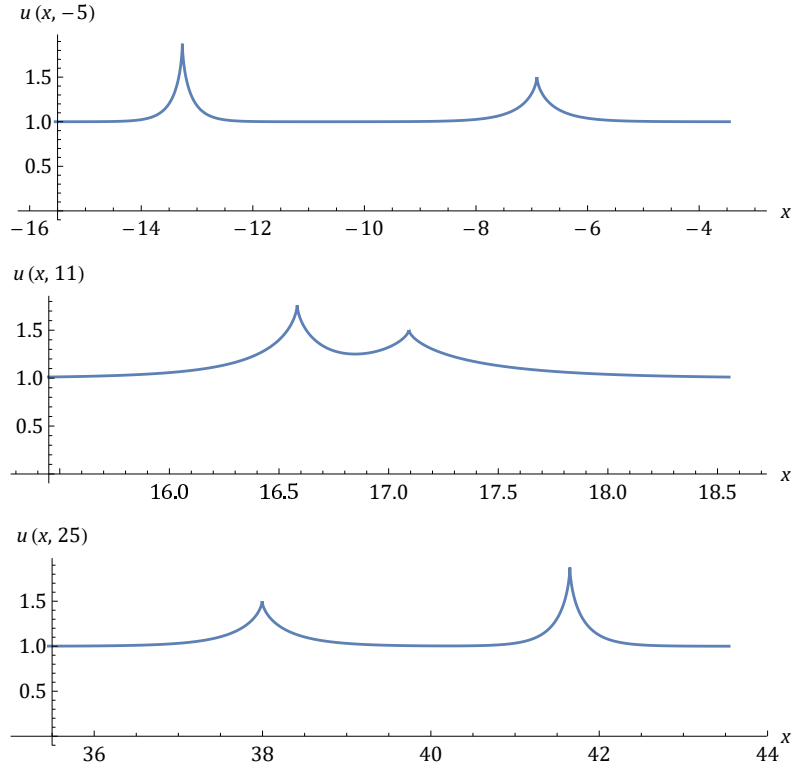}
\caption{Snapshots of a two-cuspon solution. Both cuspons are right-moving, $\lambda_1=1.0,$ $\lambda_2 = 2.0,$ $u_0=1.0, $  $\xi_1=0.0,$ $\xi_2=5.0.$}\label{fig3-1a}
\end{figure}
\begin{figure}[h!]
\centering
\includegraphics[width=0.7\textwidth]{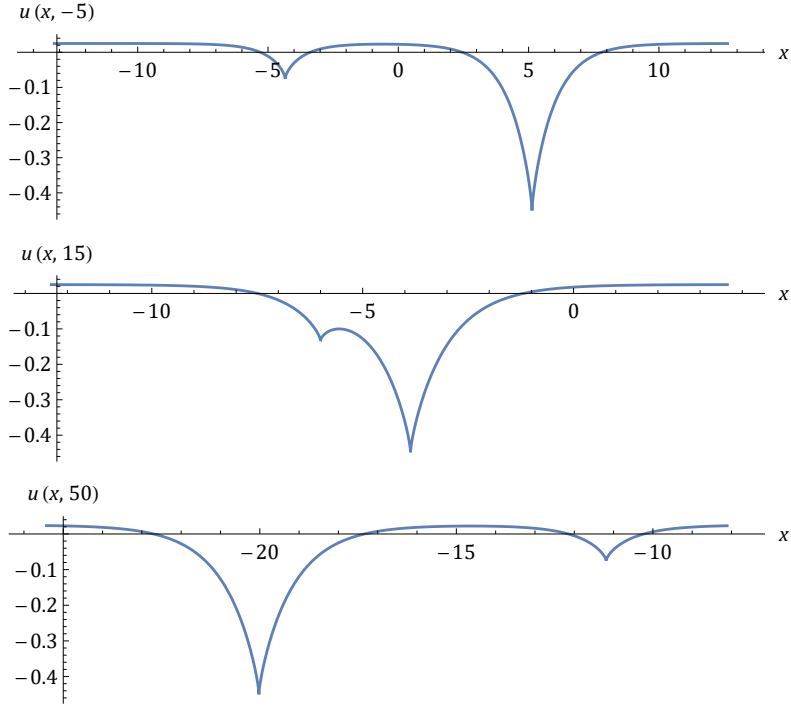}
\caption{Snapshots of a two-cuspon solution. Both cuspons are left moving, $\lambda_1=1.0,$ $\lambda_2 = 2.0,$ $u_0=0.1, $  $\xi_1=0.0,$ $\xi_2=10.0.$ }\label{fig3-1b}
\end{figure}
\begin{figure}[h!]
\centering
\includegraphics[width=0.7\textwidth]{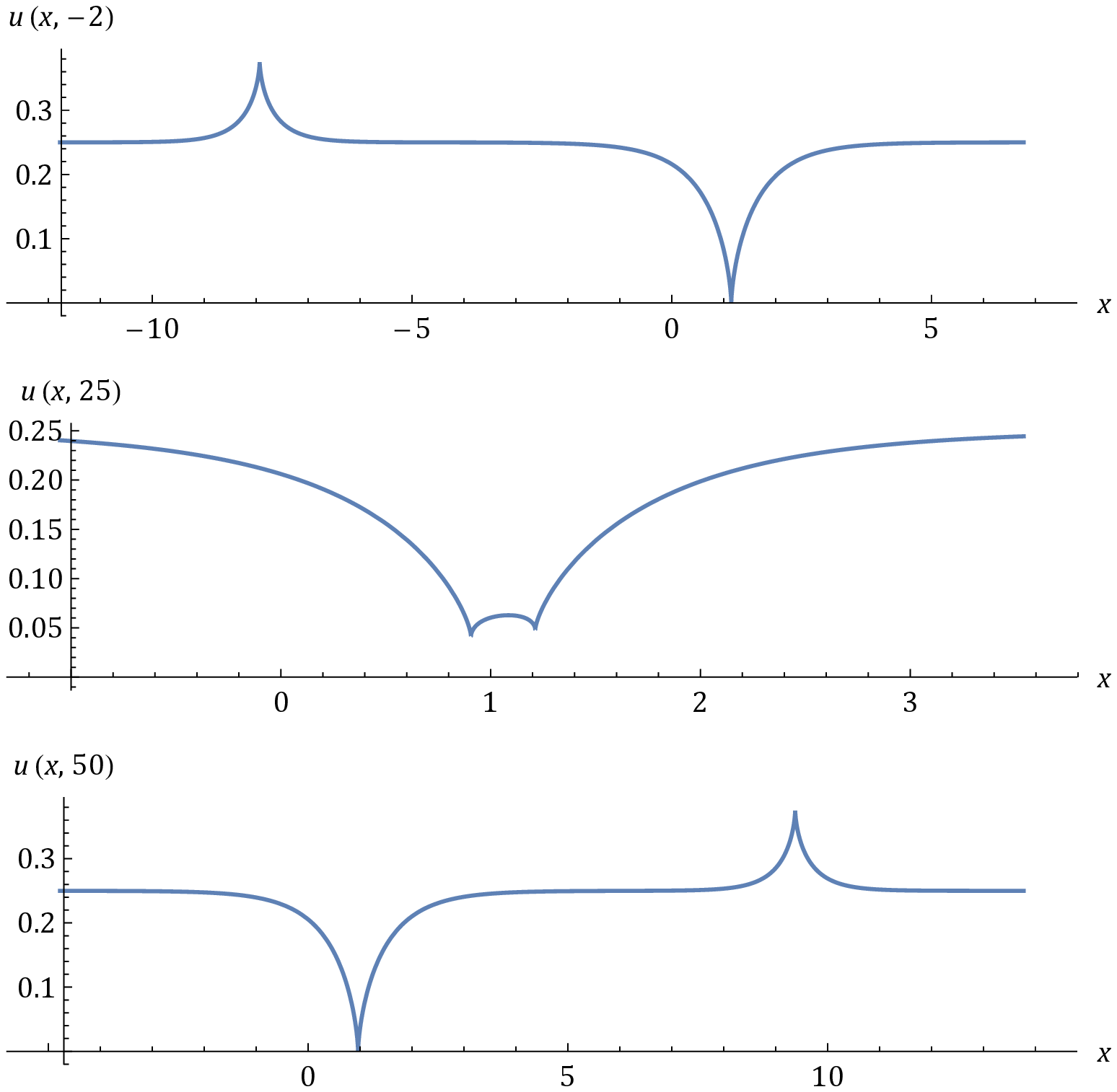}
\caption{Snapshots of a two-cuspon solution. One cuspon is left moving and the other cuspon is right moving, $\lambda_1=1.0,$ $\lambda_2 = 2.0,$ $u_0=0.25, $  $\xi_1=0.0,$ $\xi_2=10.0.$}\label{fig3-1c}
\end{figure}

\subsection{The two-soliton solution}\label{sec4.4}
As with the one-soliton solution, the dressing factor for the two-soliton solution has simple poles at the imaginary discrete eigenvalues which we denote $i\omega_1$ and $i\omega_2$, with residues $2i\omega_k A_k$ ($k=1,2$). Extending Proposition \ref{sec3.1prop1}, we have
\begin{equation}\label{sec4.4eq1}
g(y,t,\lambda) = \id+\frac{2i\omega_1 A_1 (y,t)}{\lambda-i\omega_1}+\frac{2i\omega_2 A_2 (y,t)}{\lambda-i\omega_2}.
\end{equation}
The $\mathbb{Z}_2$ reduction $ \sigma_3 \bar{g}(y,t,-\bar{\lambda})\sigma_3= g(y,t,\lambda)  $ necessitates
\begin{equation}\label{sec4.2eq2}
\sigma_3 \bar{A}_k (y,t)\sigma_3= A_k(y,t), \qquad k=1,2.
\end{equation}
The detailed computations of the two-soliton solution for the system \eqref{eq:ch} by the dressing method outlined here can be found in \cite{ILO2017}. The solution formally has the same form \eqref{solX} where this time the parameters $n_1$ and $n_2$ are given by
\begin{equation}
 n_{k}=\sqrt[4]{\frac{h_0+\Lambda_{k}}{h_0-\Lambda_k}},\quad\text{for }k=1,2.
\end{equation}
with $\Lambda_{k}=\sqrt{h_0^2-\omega_k^2}$ for $k=1, 2$. The spectral parameters $\omega_1$ and $\omega_2$ belong to the discrete spectrum of the spectral problem \eqref{ZS}, and to ensure each $\Lambda_k$ is real, we must have $\omega_k\in(0,h_0)$. Moreover, the phase of each soliton is given by
\begin{equation} \label{Omegak2}
 \Omega_k(y,t)=\Lambda_{k}\left(y-\frac{t}{2h_0}\left(u_0+\frac{1}{2\omega_k^2}\right)\right)+\xi_k\quad \text{for }k=1,2,
\end{equation}
where the constants $\xi_k$ are related to the initial separation of the solitons. The two-soliton interaction is illustrated in Figure \ref{fig2} below.
\begin{figure}[h!]
\centering
\includegraphics[scale=0.7]{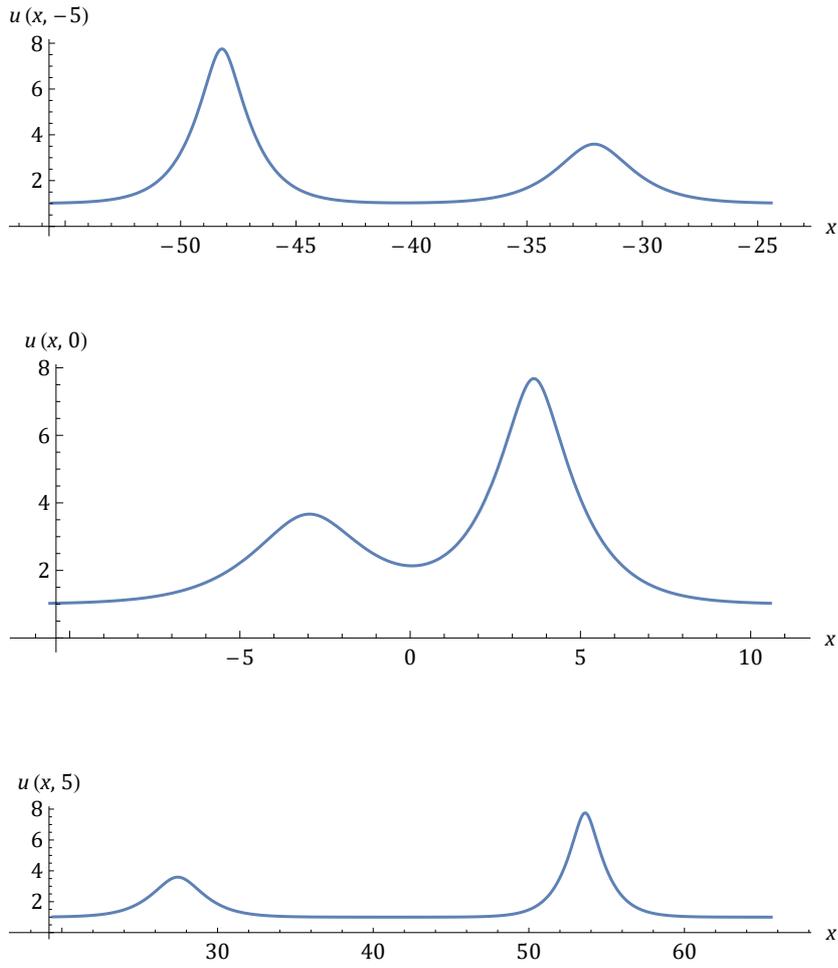}\\
\vspace{2cm}
\caption{Snapshots of the two soliton solution of the Camassa-Holm equation \eqref{eq:ch}, for three values of $t\in\{-5, 0, 5\}$. The other parameters are $u_0=1$, $\omega_1=0.35$ and $\omega_2=0.25$. The constants of integration were chosen as $\xi_1=0$ and  $\xi_2=2.0$.}\label{fig2}
\end{figure}

\subsection{The cuspon-soliton interaction}\label{sec4.5}
The dressing factor associated with the cuspon-soliton interaction has two simple poles, one imaginary pole at $\lambda=i\omega_1$ and one real pole at  $\lambda=\lambda_2$, as follows
\begin{equation}\label{sec7.2eq1}
g(y,t,\lambda) = \id+\frac{2i\omega_1 A_1 (y,t)}{\lambda-i\omega_1}+\frac{2\lambda_2 B_2 (y,t)}{\lambda-\lambda_2}.
\end{equation}
Following our previous results, the residue $B_2$ is required to be real.

The details are provided in the Appendix. The solution is formally given by the expression \eqref{solX} where $\Lambda_1=\sqrt{h_0^2-\omega_1^2}$ and $\Lambda_2=\sqrt{h_0^2+\lambda_2^2}$ along with
\begin{equation} \label{n1n2cs22}
n_1=\sqrt[4]{\frac{h_0+\Lambda_1}{h_0-\Lambda_1}},\qquad n_2=\sqrt[4]{\frac{\Lambda_2+h_0}{\Lambda_2-h_0}}.
\end{equation}
The associated phases are given by
\begin{equation} \label{Omegak77}
\begin{split}
 \Omega_1(y,t)&=\Lambda_{1}\left(y-\frac{t}{2h_0}\left(u_0+\frac{1}{2\omega_1^2}\right)\right)+\xi_1\\
 \Omega_2(y,t)&=\Lambda_{2}\left(y-\frac{t}{2h_0}\left(u_0-\frac{1}{2\lambda_2^2}\right)\right)+\xi_2.
\end{split}
\end{equation}

\begin{figure}[t!]
\centering
\includegraphics[width=0.7\textwidth,keepaspectratio]{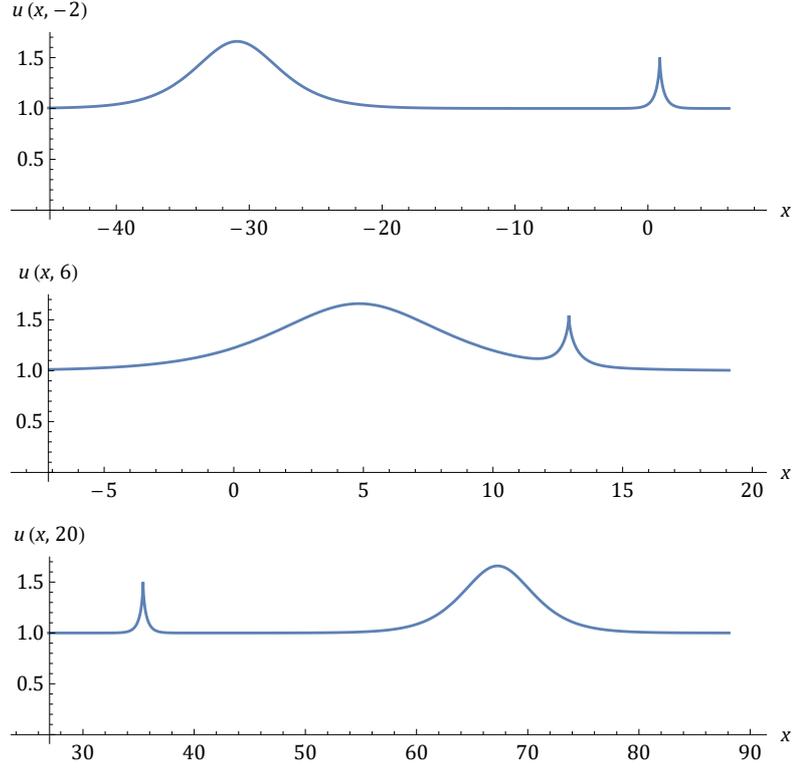}\\
\vspace{2cm}
\caption{The soliton-cuspon solution with parameters  $\lambda_1=1$, $\omega_1=0.45$, $u_0=\frac{1}{\lambda_1^2}=1.0$, $\xi_1=5$ and $\xi_2=-5$.}\label{fig4-1}
\end{figure}
\begin{figure}[t!]
\centering
\includegraphics[width=0.7\textwidth,keepaspectratio]{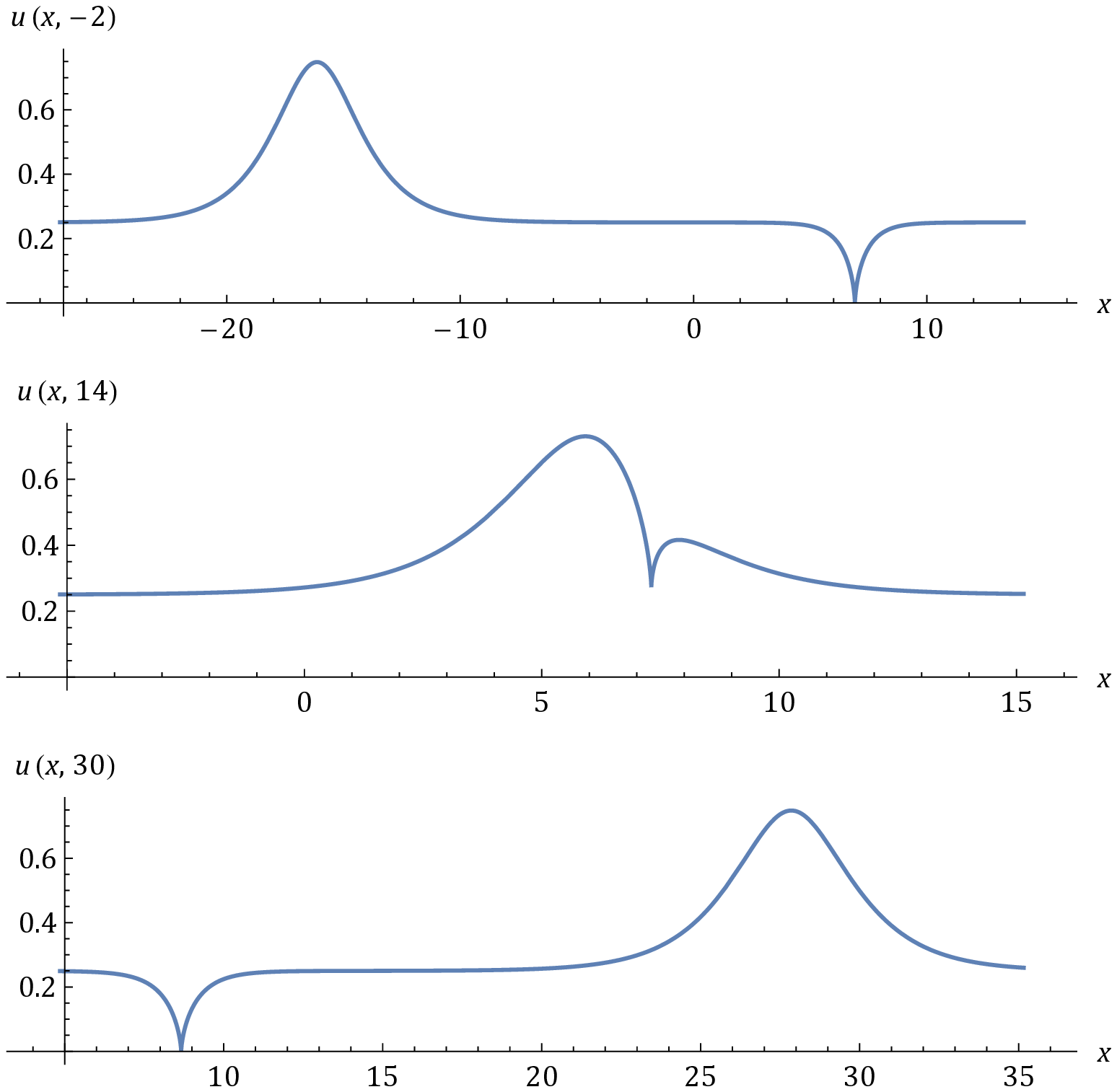}\\
\vspace{1cm}
\caption{The soliton-anticuspon solution with parameters  $\lambda_1=1$, $\omega_1=0.75$, $u_0=\frac{1}{4\lambda_1^2}=0.25$, $\xi_1=5$ and $\xi_2=-5$.}\label{fig4-2}
\end{figure}

Again, we note that the cuspon phase velocity may be either positive or negative, while the soliton phase velocity is strictly positive. Thus, the cuspon-soliton interaction may arise in two forms, namely a right moving soliton interacting with either left or right moving cuspon. Both scenarios are presented in the following: In Figure \ref{fig4-1} we present a cuspon-soliton interaction wherein both the cuspon and soliton are moving to the right. In Figure \ref{fig4-2} a cuspon-soliton interaction is shown in which the solition is right moving while the cuspon moves to the left.

\subsection{The general solution with multiple solitons and cuspons}

Now it is clear that the dressing factor for a solution with $N_1$ solitons and $N_2$ cuspons ($N=N_1+N_2$) has the form

$$ g(y,t,\lambda) = \id+\sum_{j=1}^{N_1}\frac{2i\omega_j A_j (y,t)}{\lambda-i\omega_j}+\sum_{j=1}^{N_2} \frac{2\lambda_j B_j (y,t)}{\lambda-\lambda_j}. $$
Formally these solutions are always of the form of the $ N$-soliton solution
\begin{equation} \label{XX}
X(y,t)=\frac{y}{\sqrt{u_0}}+u_0t+\ln\left \vert\frac{f_+}{f_-}\right \vert
\end{equation}
with
\begin{equation}
f_{\pm}\equiv \sum_{\sigma =0,1}\exp \left[\sum_{i=1}^{N}\sigma_i(2\Omega_i\mp\phi_i) + \!\!\sum_{1\le i<j \le N} \sigma_i \sigma_j \gamma_{i j}\right],
\end{equation}
where we introduce
\begin{equation}\label{Omega_k}
\Omega_j= \begin{cases}
&\Lambda_{j}\left(y-\frac{t}{2h_0}\left(u_0+\frac{1}{2\omega_j^2}\right)\right)+\xi_j\quad \text{for a soliton},\\
&\Lambda_{j}\left(y-\frac{t}{2h_0}\left(u_0-\frac{1}{2\lambda_j^2}\right)\right)+\xi_j \quad \text{for a cuspon}.
\end{cases}
\end{equation}
and
\begin{equation}\label{Lambda_k}
\Lambda_j=\left\{
\begin{aligned}
&\sqrt{h_0^2 - \omega_j^2}\quad \text{for a soliton},\\
&\sqrt{h_0^2+\lambda_j^2} \quad \text{for a cuspon}.
\end{aligned}
\right.
\end{equation}
and
\begin{equation}\label{n_k}
n_j=\left\{
\begin{aligned}
&\sqrt[4]{\frac{h_0+\Lambda_j}{h_0-\Lambda_j}}\quad \text{for a soliton},\\
&\sqrt[4]{\frac{\Lambda_j+h_0}{\Lambda_j-h_0}} \quad \text{for a cuspon}.
\end{aligned}
\right.
\end{equation}
along with
\begin{equation}
\phi_j = \ln (n_j)^4\qquad \gamma_{ij}= \ln\left(\frac{\Lambda_i-\Lambda_j}{\Lambda_i+\Lambda_j}\right)^2\label{gamma ij}.
\end{equation}

\begin{example}
To illustrate this generalisation we construct the soliton-cuspon-anticuspon solution as showin in Figure \ref{fig5}.
\end{example}
\begin{figure}[h!]
\centering
\includegraphics[width=0.7\textwidth]{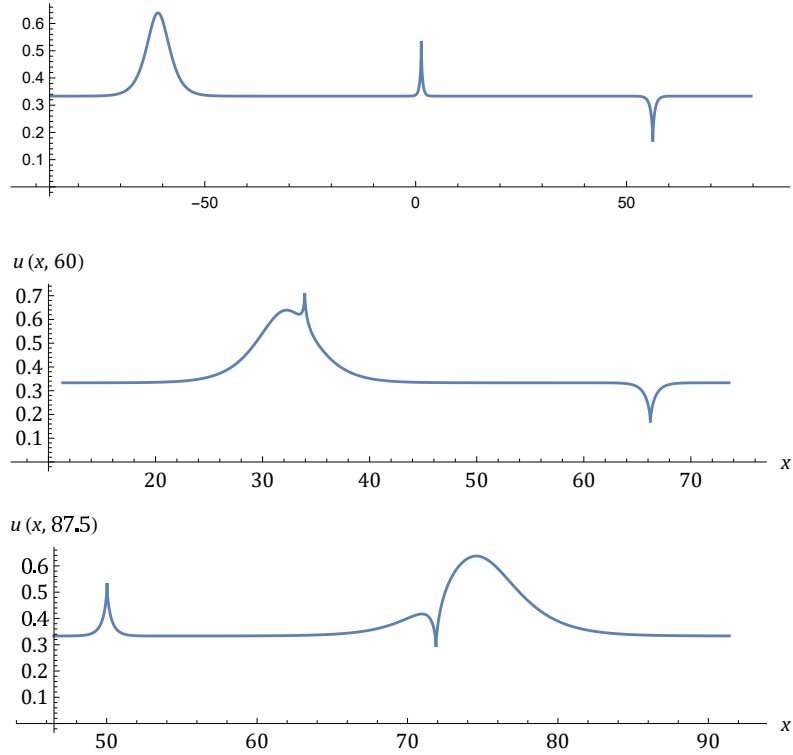}\\
\vspace{1cm}
\caption{The soliton-cuspon-anticuspon soltion with $\omega_1=0.75$, $\lambda_1=1.0$, $\lambda_2=1.95$, $u_0=\frac{1}{3\lambda_1^2}=0.33$, $\xi_1=15$, $\xi_2=-45$ and $\xi_3=-5$.}\label{fig5}
\end{figure}
This solution is explicitly construced from equation \eqref{XX} where we let $N=3$ and we find that
\begin{equation}
\begin{split}
f_+ &=  1+ \left(\sum_{k=1}^{3}\frac{1}{n_k^4 }e^{2\Omega_k}\right)+ \left(\frac{\Lambda_1-\Lambda_2}{\Lambda_1+\Lambda_2}\right)^2\frac{e^{2\Omega_1+2\Omega_2}}{n_1^4n_2^4}
\\& +
\left(\frac{\Lambda_1-\Lambda_3}{\Lambda_1+\Lambda_3}\right)^2\frac{e^{2\Omega_1+2\Omega_3}}{n_1^4n_3^4}
+
\left(\frac{\Lambda_2-\Lambda_3}{\Lambda_2+\Lambda_3}\right)^2\frac{e^{2\Omega_2+2\Omega_3}}{n_2^4n_3^4}\\
&+\left(\frac{\Lambda_1-\Lambda_2}{\Lambda_1+\Lambda_2}\right)^2\left(\frac{\Lambda_1-\Lambda_3}{\Lambda_1+\Lambda_3}\right)^2
\left(\frac{\Lambda_2-\Lambda_3}{\Lambda_2+\Lambda_3}\right)^2\frac{e^{2\Omega_1+2\Omega_2+2\Omega_3}}{n_1^4n_2^4n_3^4} \\
f_- &=  1+\left(\sum_{k=1}^{3}{n_k^4 }e^{2\Omega_k}\right) +  \left(\frac{\Lambda_1-\Lambda_2}{\Lambda_1+\Lambda_2}\right)^2 n_1^4n_2^4e^{2\Omega_1+2\Omega_2}  \\ & +
\left(\frac{\Lambda_1-\Lambda_3}{\Lambda_1+\Lambda_3}\right)^2  n_1^4n_3^4e^{2\Omega_1+2\Omega_3}
+\left(\frac{\Lambda_2-\Lambda_3}{\Lambda_2+\Lambda_3}\right)^2 n_2^4n_3^4e^{2\Omega_2+2\Omega_3} \\& +\left(\frac{\Lambda_1-\Lambda_2}{\Lambda_1+\Lambda_2}\right)^2\left(\frac{\Lambda_1-\Lambda_3}{\Lambda_1+\Lambda_3}\right)^2 \left(\frac{\Lambda_2-\Lambda_3}{\Lambda_2+\Lambda_3}\right)^2n_1^4n_2^4n_3^4e^{2\Omega_1+2\Omega_2+2\Omega_3} .
 \end{split}
\end{equation}

\section{The geometry of the soliton and cuspon solutions}
It is well known that the CH equation is closely related to the group of the diffeomorphisms of the real line, cf. \cite{HMR98,CK03}. We are going to explore this relation in some more detail, making use of the variables that we have already evaluated explicitly for the soliton and cuspon solutions.  Let us consider a local chart parameterisation of the Lie group $\mathcal{G}\simeq\mathrm{Diff}(\mathbb{R})$ given  by the coordinate $X.$ Then $dX$ is the basis for the co-tangent bundle $T^*\mathrm{Diff}(\mathbb{R})$, $\mathcal{L}=PdX$ is a 1-form, and $\omega=d \mathcal{L} = dP \wedge dX$ is the standard symplectic form. Thus, if $(X,P)$ are the canonical Hamiltonian variables, then $(dX,dP)$ are the canonical local coordinates on the phase space $T^{*}{\mathcal{G}}$.  The action of $\mathcal{G}$ in coordinate form is
\[g(t)y =X_{t}(y,t)=u(X(y,t),t)=u\circ X(y,t)=(u\circ g)y.\]
Thus $u=g_{t}g^{-1}\in\mathfrak{g}$, where $\mathfrak{g}=Vect(\mathbb{R})$, is the Lie algebra of vector fields of the form $u\partial_{x} $.

Now we recall the following result:
\begin{theorem}\cite{K81,K93}
The dual space of $\mathfrak{g}$ is a space of distributions but
the subspace of local functionals, called the regular dual
$\mathfrak{g}^*$,  is naturally identified with the space of
quadratic differentials $q(x)dx^2$ on $\mathbb{R}$. The pairing
is given for any vector field $u\partial_x\in\text{Vect}(\mathbb{R})$ by $$\langle qdx^2, u\partial_x\rangle=\int_{\mathbb{R}}q(x)u(x)dx$$. The coadjoint action coincides with the action of a diffeomorphism on the quadratic differential:
\[\text{Ad}_g^*:\quad q(y,0)dy^2\mapsto q(X,t)dX^2=q(X(y,t),t)X_y^2dy^2.\]
\end{theorem}
We therefore have
\begin{equation}\label{coad}
\begin{aligned}
\frac{\mathrm{d}}{\mathrm{d}t} Ad_{g(t)}^{*}q(0)&=\frac{\mathrm{d}}{\mathrm{d}t}  \left(X_y^2 q(X(y,t),t) \right)\\
                                                &=X_y^2(2u_X q(X,t)+u q_X+q_t)=X_y^2\left[(2u_x q+u q_x+q_t)\circ g\right] (y)=0,
\end{aligned}\end{equation}
iff $q$ satisfies the Camassa-Holm equation \eqref{eq:ch}.
In order to establish the relationship with the Hamiltonian variables, first we notice that
\[\frac{\partial u(X(y,t),t)}{\partial y}= u_X(X(y,t),t)X_y\]
and also
\[\frac{\partial u(X(y,t),t)}{\partial y}= \frac{\partial X_t}{\partial y}=X_{t y}.\]
Thus $u_X(X(y,t),t)=X_{t y}/X_y$. Similarly, $$u_{XX}(X(y,t),t)=\frac{1}{X_{y}}\left
(\frac{X_{t y}}{X_y}\right)_y$$ and $$
q(X(t,y),t)=u(X(t,y),t)-u_{XX}(X(t,y),t)=X_t-\frac{1}{X_{y}}\left
(\frac{X_{t y}}{X_y}\right)_y;$$

\begin{equation} \label{mPX}
q(x,t)=\int_{\mathbb{R}}P(y,t)\delta(x-X(y,t)) dy
\end{equation}
with
\begin{equation}\label{P(y,t)}
P(y,t)=X_t X_y-\left(\frac{X_{t y}}{X_y} \right)_y,
\end{equation}
and
\begin{equation}\label{uPX}
u(x,t)=\frac{1}{2}\int_{\mathbb{R}}  G(x-X(y,t)) P(y,t) d y.
\end{equation}
where $G(x)\equiv  \frac{1}{2}e^{-|x|}$ is the Green function of the operator $1-\partial_x^2$ and $(X(y,t),P(y,t))$ are quantities well
defined in terms of the scattering data.

With a substitution of (\ref{mPX}) and (\ref{uPX}) into the Camassa-Holm equation
(\ref{eq:ch}) and using the fact that
\begin{equation*}
f(x)\delta'(x-x_0)=f(x_0)\delta'(x-x_0)-f'(x_0)\delta(x-x_0)
\end{equation*}
we derive a system of integral equations for $X$ and $P$, namely
\begin{align}
\label{eq10}X_t(y,t)&=\int_{\mathbb{R}} G(X(y,t)-X(\underline{y},t))P(\underline{y},t)\text{d}\underline{y},\\
\label{eq11}P_t(y,t)&=-\int_{\mathbb{R}} G'(X(y,t)-X(\underline{y},t))P(y,t)P(\underline{y},t)\text{d}\underline{y}.
\end{align}
Moreover, from equations (\ref{mPX}) and (\ref{uPX}) a Hamiltonian $H_1$ can be identified
\begin{equation}\label{eq12}
H_1[X,P]=\frac{1}{2}\int_{\mathbb{R}}
G(X(y_1,t)-X(y_2,t))P(y_1,t)P(y_2,t)\text{d}y_1\text{d}y_2
\end{equation}
in which case equations (\ref{eq10}) and (\ref{eq11}) can be written as
\begin{equation}\label{eq13}
 X_t(y,t)=\frac{\delta H_1}{\delta P(y,t)}, \qquad P_t(y,t)=-\frac{\delta H_1}{\delta X(y,t)},
\end{equation}
that is to say, these equations are Hamiltonian with respect to
the canonical Poisson bracket
\begin{equation}\label{PBa}
\{A,B\}_c=\int_{\mathbb{R}}\left(\frac{\delta A}{\delta X(y,t)} \frac{\delta B}{\delta P(y,t)}- \frac{\delta B}{\delta X(y,t)}\frac{\delta A}{\delta P(y,t)}\right)\mathrm{d}y.
\end{equation}
where the canonical variables are $X(y,t)$, $P(y,t)$, with
\begin{align}
\label{Can1}\{X(y_1,t), P(y_2,t)\}_c&=\delta(y_1-y_2),\\
\label{Can2}\{P(y_1,t), P(y_2,t)\}_c&=\{X(y_1,t), X(y_2,t)\}_c=0.
\end{align}
Furthermore, using the canonical Poisson brackets (\ref{Can1}), (\ref{Can2}) and the defining integral \eqref{mPX} one can compute
\begin{equation*}
\{q(x_1,t),q(x_2,t)\}_c=-\left(q(x_1,t)\frac{\partial}{\partial x_1}+\frac{\partial}{\partial x_1} q(x_1,t)\right)\delta(x_1-x_2)\equiv\mathcal{J}_1(x_1)\delta(x_1-x_2).
\end{equation*}
Now it is straightforward to check that (\ref{eq:ch}) can be written in a Hamiltonian form as
\begin{equation*}
  q_t=\{q,H_1\}_c,
\end{equation*}
with the Poisson bracket, generated by $ \mathcal{J}_1$:
\begin{equation}
\begin{aligned}\label{PBCH}
\{A,B \}_c&=\int_{\mathbb{R}} \frac{\delta A}{\delta q(x)} \mathcal{J}_1(x) \frac{\delta B}{\delta q(x)}\text{d}x \\
          &= -\int_{\mathbb{R}} q(x) \left(\frac{\delta A}{\delta q(x)} \frac{\partial}{\partial x} \frac{\delta B}{\delta q(x)}-\frac{\delta B}{\delta q(x)}\frac{\partial}{\partial x}\frac{\delta A}{\delta q(x)} \right)\text{d}x
\end{aligned}
\end{equation}
and the Hamiltonian $H_1$ given in equation \eqref{eq12}, which can be written also as
\[H_1[q]=\frac{1}{2}\int_{\mathbb{R}} \big(q(x,t) u(x,t)-u_0^2\big) dx.\]
Thus we have an equivariant momentum map\footnote{This should not be confused with $J$ from the ZS spectral problem \eqref{sec2.2eq5}.} $J: T^*\mathcal{G}\rightarrow\mathfrak{g}^*$ for the co-adjoint action of $\mathcal{G}$ as shown on Figure \ref{MMap}.
\begin{figure}
\[
\begin{diagram}
\node{(dX,dP)_{t=0}\in T^* \mathcal{G} }\arrow[3]{s,r}{J} \arrow[7]{e,t}{g(t) \in \mathcal{G}}
          \node[10]{(dX,dP)_{t} \in T^*\mathcal{G}}    \arrow[3]{s,r}{J}
          \\
          \\
          \\
\node{q(0) \in \mathfrak{g}^*}\arrow[8]{e,t}{\text{Ad}_{g(t)}^*} \node[10]{q(t)\in \mathfrak{g}^*}
\end{diagram}
\]
\caption{Equivariant Momentum Map: quantities related to the Camassa-Holm equation. The subindex $t$ indicates that the corresponding variables are evaluated at time $t$.}
\label{MMap}
\end{figure}
This means that the values of the corresponding $\mathfrak{g}^{*}$ quantities produced by the co-adjoined action of the group $\mathcal{G}$ are conserved by the momentum map $J$ in the sense of \eqref{coad}.

\section{Discussion}
Since the emphasis of this study was on the Zakharov-Shabat dressing method many important additional questions have been overlooked - such as the phase shifts after the interaction and the peakon (antipeakon) limit when $u_0 \to 0.$  These issues have been studied previously, for instance in \cite{P07,P09,PM06}. We mention only that the cuspon behavior is very similar to the peakon behaviour, especially the peakon-antipeakon interactions.

Another interesting aspect of the momentum map obtained here is that in the peakon limit equation \eqref{mPX} becomes the well known singular momentum map
used for the construction of peakon, filament and sheet singular solutions for higher dimensional EPDiff equations  \cite{HoMa2004}.  Holm and Staley \cite{HoSt2003a} introduced the following measure-valued singular momentum solution ansatz for the $n-$dimensional solutions of the EPDiff equation
\begin{equation}\label{m-ansatz}
\mathbf{q}(\mathbf{x},t) = \sum_{a=1}^N\int\mathbf{P}^a(s,t)\,\delta\left(\,\mathbf{x}-\mathbf{Q}^a(s,t)\,\right) d s.
\end{equation}
These singular momentum solutions, called ``diffeons,'' are vector density functions supported in ${\mathbb{R}}^n$ on a set of $N$ surfaces (or curves) of  co-dimension $(n-k)$ for $s\in{\mathbb{R}}^{k}$ with $k<n$.  They may, for example, be supported
on sets of points (vector peakons, $k=0$), one-dimensional
filaments (strings, $k=1$), or two-dimensional surfaces (sheets,
$k=2$) in three dimensions. These solutions represent smooth
embeddings ${\rm Emb}(\mathbb{R}^k,\mathbb{R}^n)$ with $k<n$. In
contrast, the similar expression (\ref{mPX}) for the soliton
solutions represent smooth functions $\mathbb{R}\to\mathbb{R}$.

\section*{Acknowledgements} R.I. is grateful to Prof. D.D. Holm for many discussions on the problems treated in this paper.

\section{Appendix}
In this appendinx we provide some details on the derivation of the soliton-cuspon solution. Applying equation \eqref{sec3.1eq3} to the dressing factor $g$ as given by equation \eqref{sec7.2eq1} ensures the matrix valued residues satisfy the following
\begin{equation}\label{sec7.2eq3}
\begin{split}
&A_{1,y}+h \sigma_3 A_1 - A_1 h_0\sigma_3 - i\omega_1 [J, A_1]=0, \\
&B_{2,y}+h \sigma_3 B_2 - B_2 h_0\sigma_3 - \lambda_2 [J, B_2]=0,
\end{split}
\end{equation}
Writing the rank one matrix solutions $A_1$ and $B_2$ in the form
\begin{equation}\label{sec7.2eq4}
\begin{split}
A_1 &= \ket{n}\bra{m}, \qquad  B_2 = \ket{N}\bra{M},\\
\end{split}
\end{equation}
we deduce
\begin{equation}\label{sec7.2eq5}
\begin{cases}
\partial_y \ket{n} + (h\sigma_3-i\omega_1 J)\ket{n}=0, \qquad \partial_y\bra{m}= \bra{m} (h_0\sigma_3-i\omega_1 J)\\
\partial_y \ket{N} + (h\sigma_3-\lambda_2 J)\ket{N}=0, \qquad \partial_y\bra{M}= \bra{M} (h_0\sigma_3-\lambda_2 J).
\end{cases}
\end{equation}
The vectors  $\bra{m},\bra{M}$ satisfy the bare equations and therefore are known in principle and have been obtained previously (see sections \S\S \ref{sec4.3} -- \ref{sec4.4}).

The dressing factor \eqref{sec7.2eq1} at $\lambda=0$ is
\begin{eqnarray}\label{sec7.2eq12}
g(y,t;0)&=&\id - 2(A_1+B_2)=\mathrm{diag}(g_{11},g_{22})\\
&=&\mathrm{diag}\left(\frac{i\omega_1 M_1 m_2 -\lambda_2 M_2 m_1}{i\omega_1 M_2 m_1-\lambda_2 M_1 m_2},\frac{i\omega_1 M_2 m_1-\lambda_2 M_1 m_2}{i\omega_1 M_1 m_2 -\lambda_2 M_2 m_1}  \right), \nonumber
\end{eqnarray}
while the differential equation for $X(y,t)$ is
\begin{equation}\label{sec7.2eq13}
(\partial_y X) e^{X-2h_0 y-u_0 t} =g_{22}^2=\left(\frac{\lambda_1 M_2 m_1-\lambda_2 M_1 m_2}{\lambda_1 M_1 m_2 -\lambda_2 M_2 m_1}\right)^2.
\end{equation}
Choosing $m_1, m_2$ as per the soliton solution cf. \cite{ILO2017}, and recalling $\Lambda_1=\sqrt{h_0^2-\omega_1^2}$, we then have
\begin{equation*}
\begin{split}
  m_1&=\mu_1\sqrt{\frac{h_0+\Lambda_1}{2\Lambda_1}}e^{\Omega_1(y,t)}+\mu_2 \sqrt{\frac{h_0-\Lambda_1}{2\Lambda_1}}e^{-\Omega_1(y,t)},\\ m_2&=i\left(\mu_1 \sqrt{\frac{h_0-\Lambda_1}{2\Lambda_1}}e^{\Omega_1(y,t)}+\mu_2 \sqrt{\frac{h_0+\Lambda_1}{2\Lambda_1}}e^{-\Omega_1(y,t)}\right) ,
 \end{split}
\end{equation*}
where $\mu_k$ are positive constants. Ignoring an irrelevant overall constant of $\sqrt{\mu_1 \mu_2}(2\Lambda_1)^{-1/2}$ (see \S \ref{sec4.3}) and changing the definition of $\Omega_1(y,t)$ by an additive constant, as given by
\begin{equation}
    \Omega_1(y,t)=\Lambda_{1}\left(y-\frac{t}{2h_0}\left(u_0+\frac{1}{2\omega_1^2}\right)\right)+ \ln \sqrt{\frac{\mu_1}{\mu_2}},
\end{equation}
we obtain the simplified expressions
\begin{equation} \label{ml12}
\begin{split}
  m_1&=\sqrt{h_0+\Lambda_1}e^{\Omega_1(y,t)}+\sqrt{h_0-\Lambda_1}e^{-\Omega_1(y,t)},\\ m_2&=i\left(\sqrt{h_0-\Lambda_1}e^{\Omega_1(y,t)}+\sqrt{h_0+\Lambda_1}e^{-\Omega_1(y,t)}\right).  \end{split}
\end{equation}
Similarly, as per the cuspon solution, we define the constant vector $\bra{M_{(0)}}V_2= (\nu_1,  \nu_2)$ with $\nu_1, \nu_2$ real and positive, thereby ensuring
\begin{equation} \label{Ml12}
\begin{cases}
  &M_1=\sqrt{\Lambda_2+h_0}e^{\Omega_2(y,t)}-\sqrt{\Lambda_2-h_0}e^{-\Omega_2(y,t)},\\
  &M_2=\sqrt{\Lambda_2-h_0}e^{\Omega_2(y,t)}+\sqrt{\Lambda_2+h_0}e^{-\Omega_2(y,t)},\\
  &\Omega_2(y,t)=\Lambda_{2}\left(y-\frac{t}{2h_0}\left(u_0-\frac{1}{2\lambda_2^2}\right)\right)+ \ln \sqrt{\frac{\nu_1}{\nu_2}} , \\
  &\Lambda_2=\sqrt{h_0^2+\lambda_2^2}.
\end{cases}
\end{equation}
The expression
\begin{equation}
g_{22}=\frac{i\omega_1 M_2 m_1-\lambda_2 M_1 m_2}{i\omega_1 M_1 m_2  -\lambda_2 M_2 m_1}=\frac{\mathcal{T}_{CS}}{\mathcal{B}_{CS}},
\end{equation}
whose explicit form may be deduced frome equations equations \eqref{ml12}--\eqref{Ml12}, has denominator
\begin{equation}
\begin{split}
\mathcal{B}_{CS}=&-\omega_1 \lambda_2 \left(\frac{\Lambda_2 - \Lambda_1}{\sqrt{(h_0-\Lambda_1)(\Lambda_2-h_0)}}e^{\Omega_1+\Omega_2}+ \frac{\Lambda_2 - \Lambda_1}{\sqrt{(h_0+\Lambda_1)(\Lambda_2+h_0)}}e^{-\Omega_1-\Omega_2}  \right. \\
  & + \left.  \frac{\Lambda_1 + \Lambda_2}{\sqrt{(h_0-\Lambda_1)(\Lambda_2+h_0)}}e^{\Omega_1-\Omega_2} + \frac{\Lambda_1 + \Lambda_2}{\sqrt{(h_0+\Lambda_1)(\Lambda_2-h_0)}}e^{-\Omega_1+\Omega_2} \right),
\end{split}
\end{equation}
where we note that $\Lambda_2 > h_0 > \Lambda_1$ thus ensuring $\Lambda_2 - \Lambda_1 >0$.

Introducing the constants
\begin{equation} \label{n1n2cs}
n_1=\sqrt[4]{\frac{h_0+\Lambda_1}{h_0-\Lambda_1}},\qquad n_2=\sqrt[4]{\frac{\Lambda_2+h_0}{\Lambda_2-h_0}},
\end{equation}
we re-write this denominator according to
\begin{equation}
\mathcal{B}_{CS}=-\sqrt{\omega_1 \lambda_2} (\Lambda_1^2-\Lambda_1^2) \left(n_1 n_2\frac{e^{\Omega_1+\Omega_2}}{\Lambda_1+\Lambda_2} +\frac{1}{n_1 n_2}\frac{e^{-\Omega_1-\Omega_2}}{\Lambda_1+\Lambda_2}
+\frac{n_1}{n_2}\frac{e^{\Omega_1-\Omega_2}}{\Lambda_2-\Lambda_1}+
\frac{n_2}{n_1}\frac{e^{-\Omega_1+\Omega_2}}{\Lambda_2-\Lambda_1} \right).
\end{equation}
Similarly it is found that the numerator assumes the form
\begin{equation}
\mathcal{T}_{CS}=i\sqrt{\omega_1 \lambda_2} (\Lambda_1^2-\Lambda_2^2) \left(\frac{1}{n_1 n_2}\frac{e^{\Omega_1+\Omega_2}}{\Lambda_1+\Lambda_2} -n_1 n_2\frac{e^{-\Omega_1-\Omega_2}}{\Lambda_1+\Lambda_2}
-\frac{n_2}{n_1}\frac{e^{\Omega_1-\Omega_2}}{\Lambda_2-\Lambda_1}+
\frac{n_1}{n_2}\frac{e^{-\Omega_1+\Omega_2}}{\Lambda_2-\Lambda_1} \right).
\end{equation}
As with the two-cupson solutions we seek a solution of  equation \eqref{sec4.2eq13} in the form of equation \eqref{sec4.1eq7}  with
\begin{equation}
\mathcal{A}_{CS}=\alpha_1 e^{\Omega_1+\Omega_2}+\alpha_2 e^{-\Omega_1-\Omega_2}+\alpha_3 e^{\Omega_1-\Omega_2}+\alpha_4 e^{-\Omega_1+\Omega_2},
\end{equation}
where the constants $\left\{\alpha_l\right\}_{l=1}^{4}$ are as yet unknown. Equation \eqref{sec4.2eq13} ensures that
\begin{equation}
2h_0\mathcal{A}_{CS}\mathcal{B}_{CS}+\mathcal{B}_{CS}\partial_y\mathcal{A}_{CS}-\mathcal{A}_{CS}\partial_y\mathcal{A}_{CS}=2h_0\mathcal{T}_{CS}^2
\end{equation}
has a solution  given by
\begin{equation}
\mathcal{A}_{CS}=\sqrt{\omega_1 \lambda_2} (\Lambda_2^2-\Lambda_1^2) \left(\frac{1}{n_1^3 n_2^3}\frac{e^{\Omega_1+\Omega_2}}{\Lambda_1+\Lambda_2}+n_1^3 n_2^3\frac{e^{-\Omega_1-\Omega_2}}{\Lambda_1+\Lambda_2}
+\frac{n_2^3}{n_1^3}\frac{e^{\Omega_1-\Omega_2}}{\Lambda_2-\Lambda_1}+
\frac{n_1^3}{n_2^3}\frac{e^{-\Omega_1+\Omega_2}}{\Lambda_2-\Lambda_1} \right).
\end{equation}
The ratio $\mathcal{A}_{CS}/\mathcal{B}_{CS}$ may now be written as
\begin{equation} \label{Acs}
\frac{\mathcal{A}_{CS}}{\mathcal{B}_{CS}}=n_1^4 n_2^4\frac{1+ \frac{1}{n_1^6 }\frac{\Lambda_1+\Lambda_2}{\Lambda_2-\Lambda_1}e^{2\Omega_1}+ \frac{1}{n_2^6 }\frac{\Lambda_1+\Lambda_2}{\Lambda_2-\Lambda_1}e^{2\Omega_2} +\frac{e^{2\Omega_1+2\Omega_2}}{n_1^6n_2^6} }{1+ n_1^2\frac{\Lambda_1+\Lambda_2}{\Lambda_2-\Lambda_1}e^{2\Omega_1}+ n_2^2 \frac{\Lambda_1+\Lambda_2}{\Lambda_2-\Lambda_1}e^{2\Omega_2}+ n_1^2n_2^2e^{2\Omega_1+2\Omega_2} }
\end{equation}
which we simplify by means of the following re-definitions:
\begin{equation} \label{Omega12}
\begin{split}
 \Omega_1(y,t)&=\Lambda_{1}\left(y-\frac{t}{2h_0}\left(u_0+\frac{1}{2\omega_1^2}\right)\right)+ \ln \sqrt{\frac{\mu_1}{\mu_2}}-\ln n_1 + \frac{1}{2}\ln \frac{\Lambda_1+\Lambda_2}b{\Lambda_2 - \Lambda_1},\\
 \Omega_2(y,t)&=\Lambda_{2}\left(y-\frac{t}{2h_0}\left(u_0-\frac{1}{2\lambda_2^2}\right)\right)+ \ln \sqrt{\frac{\nu_1}{\nu_2}}-\ln n_2 + \frac{1}{2}\ln \frac{\Lambda_1+\Lambda_2}{\Lambda_2 - \Lambda_1}.
 \end{split}
\end{equation}
These re-dfinitions are valid since $\Lambda_2> \Lambda_1$, as was previously noted.
Alternatively, these may be simply written as
\begin{equation} \label{Omegak7}
\begin{split}
 \Omega_1(y,t)&=\Lambda_{1}\left(y-\frac{t}{2h_0}\left(u_0+\frac{1}{2\omega_1^2}\right)\right)+\xi_1\\
 \Omega_2(y,t)&=\Lambda_{2}\left(y-\frac{t}{2h_0}\left(u_0-\frac{1}{2\lambda_2^2}\right)\right)+\xi_2
\end{split}
\end{equation}
for some constants $\left\{\xi_k\right\}_{k=1}^{2}$ related to the initial separation of the cuspon and soliton. This allows the expression \eqref{Acs} to be written as \begin{equation}
\frac{\mathcal{A}_{CS}}{\mathcal{B}_{CS}}=n_1^4 n_2^4\frac{1+ \frac{1}{n_1^4 }e^{2\Omega_1}+ \frac{1}{n_2^4 }e^{2\Omega_2} +\left(\frac{\Lambda_1-\Lambda_2}{\Lambda_1+\Lambda_2}\right)^2\frac{e^{2\Omega_1+2\Omega_2}}{n_1^4n_2^4} }{1+ n_1^4e^{2\Omega_1}+ n_2^4e^{2\Omega_2} + \left(\frac{\Lambda_1-\Lambda_2}{\Lambda_1+\Lambda_2}\right)^2n_1^4n_2^4e^{2\Omega_1+2\Omega_2} }
\end{equation}
with
\[X(y,t)=\frac{y}{\sqrt{u_0}}+u_0t+\ln\left \vert\frac{\mathcal{A}_{CS}}{\mathcal{B}_{CS}}\right\vert.\]

\end{document}